\documentclass[11pt]{article}

\usepackage{amsmath}
\usepackage{amssymb}

\usepackage{graphicx}

\usepackage{cite}

\usepackage{color} 

\usepackage{lscape}
\usepackage{tabularx}
\usepackage{booktabs}

\newtheorem{assumption}{Assumption}
\newtheorem{problem}{Problem}
\newtheorem{definition}{Definition}
\newtheorem{theorem}{Theorem}
\newtheorem{proposition}{Proposition}
\newtheorem{corollary}{Corollary}
\newtheorem{example}{Example}
\newtheorem{estimator}{Estimator}
\def\enstate{\hfill\ensuremath{\Box}}
\def\enproof{\hfill\ensuremath{\blacksquare}}
\def\argmax{\mathop{\mathrm{arg\ max}}}
\def\R{\ensuremath{\mathbb{R}}}
\def\Z{\ensuremath{\mathbb{Z}}}
\def\N{\ensuremath{\mathbb{N}}}
\def\centroidfold{\texttt{CentroidFold}}
\def\centroidalifold{\texttt{CentroidAlifold}}
\def\centroidhomfold{\texttt{CentroidHomfold}}
\def\centroidalign{\texttt{CentroidAlign}}
\def\contrafold{\texttt{CONTRAfold}}
\def\petfold{\texttt{PETfold}}
\def\rnaalifold{\texttt{RNAalifold}}
\def\pfold{\texttt{Pfold}}
\def\last{\texttt{LAST}}

\date{}

\title{Generalized Centroid Estimators in Bioinformatics\footnote{%
This is a corrected version of the published paper: {\bf \em PLoS ONE} 6(2):e16450, 2011. The original version is available from
http://www.plosone.org/article/info:doi/10.1371/journal.pone.0016450. Note that there are several typos in the original
version which is not in the accepted manuscript (we had no chance for proof reading of the published paper).
}}
\author{Michiaki Hamada$^{1,2}$\footnote{To whom correspondence should be addressed.
    Tel.: +81-3-5281-5271; Fax: +81-3-5281-5331; E-mail: mhamada@k.u-tokyo.ac.jp}, Hisanori Kiryu$^{1}$, Wataru Iwasaki$^{1}$, Kiyoshi Asai$^{1,2}$\\
\\
  $^{1}$the University of Tokyo, 
  $^{2}$CBRC/AIST}

\begin{document}

\maketitle



\begin{abstract}

In a number of estimation problems in bioinformatics, accuracy measures of the target problem are usually given, and it is important to design estimators that are suitable to those accuracy measures. However, there is often a discrepancy between an employed estimator and a given accuracy measure of the problem. 
In this study, we introduce a general class of efficient estimators for estimation problems on high-dimensional binary spaces, which represent many fundamental problems in bioinformatics. Theoretical analysis reveals that the proposed estimators generally 
fit with commonly-used accuracy measures (e.g. sensitivity, PPV, MCC and F-score) as well as it can be computed efficiently in many cases, and cover a wide range of problems in bioinformatics from the viewpoint of the principle of maximum expected accuracy (MEA). It is also shown that some important algorithms in bioinformatics can be interpreted in a unified manner.
Not only the concept presented in this paper gives a useful framework to design MEA-based estimators but also it is highly extendable and sheds new light on many problems in bioinformatics.
\end{abstract}



\tableofcontents

\section{Introduction}

In estimation problems in bioinformatics, the space of solutions is generally large 
and often high-dimensional.
Among them, a number of fundamental problems in bioinformatics,
such as alignment of biological sequences,
prediction of secondary structures of RNA sequences,
prediction of biological networks,
and estimation of phylogenetic trees,
are classified into estimation problems
whose solutions are in a high-dimensional binary space.
Such problems are generally difficult to solve,
and the estimates are often unreliable.

The popular solutions for these problems,
such as for the secondary structure of RNA with minimum free energy,
are the maximum likelihood (ML) estimators.
The ML estimator maximizes
the probability that the estimator is exactly correct,
but that probability is generally very small.
Noticing the drawbacks of the ML estimators,
Carvalho and Lawrence have proposed the {\em centroid estimator},
which represents an ensemble of all the possible solutions
and minimizes the expected Hamming loss of the prediction~\cite{centroid}.

In this paper,
we conduct a theoretical analysis of estimation problems
in high-dimensional binary space,
and present examples and solutions in bioinformatics.
The theories in this paper provide a unified framework for 
designing superior estimators for estimation problems in bioinformatics.
The estimators discussed in this paper,
including the ML estimator and the centroid estimator,
are formalized as maximum expected gain  (MEG) estimators,
which maximize the estimator-specific gain functions
with respect to the given probability distribution.
The objective of the estimation
is not always to find the exact solution with an extremely small probability
or to find the solution with the minimum Hamming loss,
but rather to find the most accurate estimator.
Therefore, we adopt the principle of maximum expected accuracy (MEA),
which has been successfully applied to various problems in bioinformatics,
such as the alignment of biological sequences
\cite{pmid19478997,pmid18796475, Frith_BMCB},
the secondary structure prediction of RNA 
\cite{pmid16873527,pmid19703939,pmid17182698,pmid18836192}
and other applications \cite{pmid15961464,RNAint,pmid18096039}.

Theoretical analysis, however, shows that those MEA estimators are not always robust with respect to accuracy measures.
To address this, we previously proposed the $\gamma$-centroid estimator in a few specific problems
\cite{centroidfold-submit,Frith_BMCB}.
In this paper, in order to make the $\gamma$-centroid estimator easily applicable to other estimation problems,
we introduce {an abstract form of} the $\gamma$-centroid estimator,
which is defined on general binary spaces and designed to fit to the commonly used accuracy measures.
The $\gamma$-centroid estimator is a generalization of the centroid estimator,
and offers a more robust framework for estimators than the previous estimators.
We extend the theory of maximum expected gain (MEG) estimators and $\gamma$-centroid estimators
for two advanced problems:
the estimators that represent the common solutions for multiple entries,
and the estimators for marginalized probability distributions.

\section{Materials and Methods}
%

%
\begin{problem}[Pairwise alignment of two biological sequences]\label{prob:align}
Given a pair of biological (DNA, RNA, protein) sequences $x$ and $x'$,
predict their alignment
as a point in $\mathcal{A}(x,x')$, the space of all the possible alignments of $x$ and $x'$.
\end{problem}
%
\begin{problem}[Prediction of secondary structures of RNA sequences]\label{prob:rnas}
Given an RNA sequence $x$,
predict its secondary structure
as a point in $\mathcal{S}(x)$, the space of all the possible secondary structures of $x$.
\end{problem}

A point in $\mathcal{A}(x,x')$, 
can be represented as a binary vector of $|x||x'|$ dimensions
by denoting the aligned bases across the two sequences as ''1''
and the remaining pairs of bases as ''0''.
A point in $\mathcal{S}(x)$
can also be represented as a binary vector of $|x|(|x|-1)/2$ dimensions,
which represent all the pairs of the base positions in $x$,
by denoting the base pairs in the secondary structures as ''1''.
In each problem, the predictive space ($\mathcal{A}(x,x')$ or $\mathcal{S}(x)$)
is a subset of a binary space ($\{0,1\}^{|x||x'|}$ or $\{0,1\}^{|x|(|x|-1)/2}$)
because the combinations of aligned bases or base pairs are restricted
(see 
``Discrete (binary) spaces in bioinformatics'' (Section~\ref{sec:discrete_space})
in Appendices for more formal definitions).
Therefore, Problem~\ref{prob:align} and Problem~\ref{prob:rnas}
are special cases of the following more general problem:

%
\begin{problem}[Estimation problem on a binary space]\label{prob:1}
Given a data set $D$ and a {\em predictive space} $Y$ (a set of all candidates of a prediction), 
which is a subset of $n$-dimensional binary vectors $\{0,1\}^n$, 
that is, $Y \subset \{0,1\}^n$,
predict a point $y$ in the predictive space $Y$. 
\end{problem}

Not only Problem~\ref{prob:align} and Problem~\ref{prob:rnas} 
but also a number of other problems in bioinformatics are formulated
as Problem~\ref{prob:1},
including the prediction of biological networks and the estimation of phylogenetic trees
(Problem~\ref{prob:pt}).

To discuss the stochastic character of the estimators,
the following assumption is introduced.

%
\begin{assumption}[Existence of probability distribution]\label{as:cl}
In Problem~\ref{prob:1}, 
there exists a probability distribution $p(y|D)$ on 
the predictive space $Y$.
\end{assumption}

For Problem~\ref{prob:1} with Assumption~\ref{as:cl}, 
we have the following {Bayesian maximum likelihood (ML) estimator}.
%
\begin{definition}[Bayesian ML estimator \cite{centroid}]\label{def:ML_estimator}
For Problem~\ref{prob:1} with Assumption~\ref{as:cl}, 
the estimator
\begin{align*}
\hat y^{(ML)}= \argmax_{y\in Y} p(y|D),
\end{align*}
which maximizes the Bayesian posterior probability $p(y|D)$, is referred to as a Bayesian maximum likelihood (ML) estimator. 
\end{definition}

For problems classified as Problem~\ref{prob:1}, Bayesian ML estimators have dominated the field of estimators in bioinformatics for years.
The classical solutions of Problem~\ref{prob:align} and Problem~\ref{prob:rnas}
are regarded as Bayesian ML estimators with specific probability distributions,
as seen in the following examples.
\begin{example}[Pairwise alignment with maximum score]\label{rem:align_ml}
In Problem~\ref{prob:align} with a scoring model (e.g., gap costs and a substitution matrix),
the distribution $p(y|D)$ in Assumption~\ref{as:cl}
is derived from the Miyazawa model \cite{pmid8771180}
(See ``Probability distributions $p^{(a)}(\theta|x,x')$ on $\mathcal{A}(x,x')$'' (Section~\ref{sec:P_a}) in Appendices),
and the Bayesian ML estimator is equivalent
to the alignment that has the highest similarity score. 
\end{example}
\begin{example}[RNA structure with minimum free energy]\label{rem:sec_ml}
In Problem~\ref{prob:rnas} with a McCaskill energy model \cite{pmid1695107},
the distribution $p(y|D)$ in Assumption~\ref{as:cl}
can be obtained with the aid of thermodynamics (See 
``Probability distributions $p^{(s)}(\theta|x)$ on $\mathcal{S}(x)$'' (Section~\ref{sec:P_s}) in Appendices for details), 
and the Bayesian ML estimator is equivalent
to the secondary structure that has the minimum free energy (MFE).
\end{example}

When a stochastic model such as
a pair hidden Markov model (pair HMM) in Problem~\ref{prob:align}
or a stochastic context-free grammar (SCFG) in Problem~\ref{prob:rnas} 
is assumed in such problems,
the distribution and the ML estimator are derived in a more direct manner.

The Bayesian ML estimator regards the solution which has the highest probability
as the most likely one.
To provide more general criteria for good estimators,
here we define the {\em gain function}
that gives the gain for the prediction,
and the {\em maximum expected gain (MEG) estimator} that maximizes the {\em expected gain}.

\begin{definition}[Gain function]\label{def:GainFunction}
In Problem~\ref{prob:1},
for a point $\theta \in Y$ and its prediction $y\in Y$,
a {\em gain function} is defined as
$G:Y \times Y \to \R^+$, $G(\theta,y)$.
\end{definition}

\begin{definition}[MEG estimator]\label{def:MEG}
In Problem~\ref{prob:1} with Assumption~\ref{as:cl},
the {\em maximum expected gain (MEG) estimator} is defined as
\begin{equation*}
\hat y^{(MEG)} 
= \argmax_{y \in Y} \int G(\theta, y) p(\theta|D)d\theta.
\end{equation*}
\end{definition}

If the gain function is designed according to the {\em accuracy measures} 
of the target problem,
the MEG estimator is considered as 
the maximum expected accuracy (MEA) estimator,
which has been successfully applied in bioinformatics (e.g., \cite{pmid15961464}).
Although in estimation theory a {\em loss function} that should be minimized is often used, in order to facilitate the understanding of the relationship with the MEA,
in this paper, we use a {\em gain function} that should be maximized.

The MEG estimator for the gain function $\delta(y, \theta)$ 
is the ML estimator.
Although this means that the ML estimator maximizes the probability
that the estimator is identical to the true value,
there is an extensive collection of suboptimal solutions and
the probability of the ML estimator is extremely small
in cases where $n$ in Problem~\ref{prob:1} is large.
Against this background, Carvalho and Lawrence proposed the {\em centroid estimator},
which takes into account the overall ensemble of solutions \cite{centroid}.
The centroid estimator can be defined as an MEG estimator
for a {\em pointwise gain function} as follows:

%
\begin{definition}[Pointwise gain function]\label{def:pgain_function}
In Problem~\ref{prob:1},
for a point $\theta \in Y$ and its prediction $y =\{y_i\}_{i=1}^n \in Y$,
a gain function $G(\theta, y)$ written as 
\begin{equation}
G(\theta, y)=\sum_{i=1}^n F_i(\theta, y_i) \label{eq:gain_function},
\end{equation}
where $F_i:Y\times\{0,1\}\to\R^+$ ($i=1,2,\ldots,n$), is referred to as
a {pointwise gain function}. 
\end{definition}
%

%
\begin{definition}[Centroid estimator \cite{centroid}]
In Problem~\ref{prob:1} with Assumption~\ref{as:cl},
a {\em centroid estimator} is defined
as an MEG estimator
for the pointwise gain function
given in Eq.~(\ref{eq:gain_function})
by defining $F_i(\theta,y_i)=I(\theta_i=1)I(y_i=1)+I(\theta_i=0)I(y_i=0)$.
\end{definition}

Throughout this paper,
$I(\cdot)$ is the indicator function that takes a value of 1 or 0 depending 
on whether the condition constituting its argument is true or false. 
The centroid estimator is equivalent to the
expected Hamming loss minimizer \cite{centroid}.
If we can {maximize} the pointwise gain function independently in each dimension,
we can obtain the following {\em consensus estimator},
which can be easily computed.

%
\begin{definition}[Consensus estimator \cite{centroid}]\label{def:cons}
In Problem~\ref{prob:1} with Assumption~\ref{as:cl},
the {\em consensus estimator} $\hat y^{(c)}=\{\hat y^{(c)}_i\}_{i=1}^n$
for a pointwise gain function 
is defined as
\begin{equation*}
\hat y^{(c)}_i = \argmax_{y_i \in \{0,1\}} E_{\theta|D}\left[F_i(\theta,y_i)\right] 
= \argmax_{y_i \in \{0,1\}} \int F_i(\theta,y_i)p(\theta|D)d\theta.
\end{equation*}
\end{definition}

The consensus estimator is generally {\em not} contained within the predictive space $Y$
since the predictive space $Y$ usually has {complex constraints} for each
dimension (see ``Discrete (binary) spaces in bioinformatics'' (Section~\ref{sec:discrete_space}) in Appendices). 
Carvalho and Lawrence proved a sufficient condition
for the centroid estimator to contain the consensus estimator
(Theorem~2 in \cite{centroid}).
Here, we present a more general result,
namely, a sufficient condition
for the MEG estimator for a pointwise function
to contain the consensus estimator.

\begin{theorem}\label{theorem:general}
In Problem~\ref{prob:1} with Assumption~\ref{as:cl} 
and a pointwise gain function,
let us suppose that a predictive space $Y$ can be written as
\begin{equation}
Y = \bigcap_{k=1}^K C_k,\label{eq:constrain_Y}
\end{equation}
where $C_k$ is defined as
\begin{equation*}
C_k =\biggl\{ y \in \{0,1\}^n\bigg|\sum_{i \in I_k} y_i\le 1\biggr\}
\mbox{ for } k=1,2,\ldots,K\label{eq:constrain}
\end{equation*}
for an index-set $I_k \subset \{1,2,\ldots,n\}$.
If the pointwise gain function in Eq.~(\ref{eq:gain_function}) satisfies the condition
\begin{equation}
F_i(\theta,1)-F_i(\theta,0)+F_j(\theta,1)-F_j(\theta,0) \le 0\label{eq:zyouken}
\end{equation}
for every $\theta\in Y$ and every $i,j \in I_k$ ($1\le k \le K$),
then the consensus estimator is in the predictive space $Y$,
and hence the MEG estimator contains the consensus estimator.
\end{theorem}
The above conditions are frequently satisfied in bioinformatics problems
(see Supplementary Sections~\ref{sec:discrete_space} for examples).

\section{Results}

\subsection{$\gamma$-centroid estimator: generalized centroid estimator}\label{sec:generalized}
%
In Problem~\ref{prob:1},
the ``1''s and the ``0''s in the binary vector of a prediction $y$
can be interpreted as positive and negative predictions, respectively.
The respective numbers of true positives (TP), true negatives (TN),
false positives (FP) and false negatives (FN)
for a point $\theta$ and its prediction $y$ are denoted by
$\mbox{TP}(\theta,y)$, $\mbox{TN}(\theta,y)$, $\mbox{FP}(\theta,y)$ and $\mbox{FN}(\theta,y)$,
respectively (See also Eqs~(\ref{eq:tp})--(\ref{eq:fn})).


To design a {\em superior} MEG estimator,
it is natural to use a gain function
of the following form, which yields positive scores for the number of true predictions (TP and TN)
and negative scores for those of false predictions (FP and FN):
\begin{align}
 G(\theta,y)
=\alpha_1\mbox{TP}(\theta,y)+\alpha_2 \mbox{TN}(\theta,y) 
-\alpha_3 \mbox{FP}(\theta,y)-\alpha_4 \mbox{FN}(\theta,y), \label{gain:max_lin_acc}
\end{align}
where $\alpha_k$ is a positive constant ($k=1,2,3,4$).
Note that this gain function is a pointwise gain function.

This gain function is naturally compatible with commonly used accuracy measures such as
sensitivity, PPV, MCC and F-score (a function of TP, TN, FP and FN; see ``Evaluation measures defined using TP, TN, FP and FN'' (Section~\ref{sec:acc_mes}) in Appendices for definitions).
The following Definition~\ref{def:GCE} and Theorem~\ref{thm:impl_g_est}
characterize the MEG estimator for this gain function. 

\begin{definition}[$\gamma$-centroid estimator]\label{def:GCE}
In Problem~\ref{prob:1}
with Assumption~\ref{as:cl}
and a fixed $\gamma\ge 0$,
the {$\gamma$-centroid estimator} is defined
as the MEG estimator
for the pointwise gain function
given in Eq.~(\ref{eq:gain_function}) by
\begin{align}
F_i(\theta,y_i)=I(\theta_i=0)I(y_i=0)+\gamma I(\theta_i=1)I(y_i=1).
\end{align}
\end{definition}
%
\begin{theorem}\label{thm:impl_g_est}
The MEG estimator for the gain function in Eq.~(\ref{gain:max_lin_acc})
is equivalent to a $\gamma$-centroid estimator 
with $\gamma=\frac{\alpha_1+\alpha_4}{\alpha_2+\alpha_3}$.
\end{theorem}

Theorem~\ref{thm:impl_g_est} 
(see Section~\ref{sec:proof_impl_g_est} for a formal proof)
is derived from the following relations:
\begin{align*}
&TP+FN=\sum_{i}I(\theta_{i}=1) \mbox{ and } TN+FP=\sum_{i}I(\theta_{i}=0).
\end{align*}
The $\gamma$-centroid estimator
maximizes the expected value of $TN + \gamma TP$,
and includes the centroid estimator
as a special case where $\gamma=1$.
The parameter $\gamma$ adjusts
the balance between
the gain from true negatives and that from true positives. 

The expected value of the gain function of the $\gamma$-centroid estimator is computed as follows
(see Appendices for the derivation):
\begin{align}
\sum_{i=1}^n\left[(\gamma+1)p_i-1\right] I(y_i=1) + \sum_{i=1}^n(1-p_i)\label{eq:exp_gain}
\end{align}
where
\begin{align}
p_i=p(\theta_i=1|D)=\sum_{\theta\in\Theta} I(\theta_i=1)p(\theta|D).\label{eq:marginal_form}
\end{align}

Since the second term in Eq.~(\ref{eq:exp_gain}) does not depend on $y$,
the $\gamma$-centroid estimator maximizes the first term.
The following theorem is obtained
by assuming the additional condition described below.

%
\begin{theorem}\label{theorem:g_centroid_th}
In Problem~\ref{prob:1} with Assumption~\ref{as:cl},
the predictive space $Y$ satisfies the following condition: 
if $y=\{y_i\}\in Y$, then $y'=\{y'_i\}\in Y$ where $y'_i \in \{y_i,0\}$ for all $i$.
Then, 
the $\gamma$-centroid estimator is equivalent to the estimator that maximizes
the sum of marginalized probabilities $p_i$ that are greater than ${1}/(\gamma+1)$
in the prediction.
\end{theorem}

The condition is necessary
to obtain $0$ for the $i$ that produces negative values
in the first term in Eq.~(\ref{eq:exp_gain}).
Problem~\ref{prob:rnas}, Problem~\ref{prob:align},
and many other typical problems in bioinformatics satisfy this condition.
Because the pointwise gain function of the $\gamma$-centroid estimator
satisfies Eq.~(\ref{eq:zyouken}) in Theorem~\ref{theorem:general},
we can prove the following Corollary~\ref{cor:WC_01}.
%
\begin{corollary}[$\gamma$-centroid estimator for $0\le \gamma \le 1$]\label{cor:WC_01}
In Problem~\ref{prob:1} with Assumption~\ref{as:cl}, 
the predictive space $Y$ is given in the same form 
in Eq.~(\ref{eq:constrain_Y}) of Theorem~\ref{theorem:general}.
Then, the $\gamma$-centroid estimator for $\gamma \in [0,1]$ contains 
its consensus estimator.
Moreover, the consensus estimator is identical to the following estimator
$y^*=\{y_i^*\}$:
\begin{equation}
y_i^* =\left\{
\begin{array}{ll}
1 & \mbox{if } p_i>\frac{1}{\gamma+1}\\
0 & \mbox{if } p_i\le \frac{1}{\gamma+1}
\end{array}\right.
\mbox{ for } i=1,2,\ldots,n \label{eq:gamma_consensus}
\end{equation}
where
$p_i = p(\theta_i = 1|D)=I(\theta_i=1)p(\theta|D)$.
\end{corollary}
Here, $p_i$ is the marginalized probability of the distribution
for the $i$-th dimension of the predictive space.
In Problem~\ref{prob:align},
it is known as the alignment probability,
which is defined as the probability of each pair of positions across the two sequences
being aligned.
In Problem~\ref{prob:rnas},
it is known as the base pairing probability,
which is defined as the probability of each pair of positions
forming a base pair in the secondary structure.
These marginalized probabilities can be calculated by using dynamic programming algorithms,
such as the forward-backward algorithm and the McCaskill algorithm,
depending on the model of the distributions.
(see ``Probability distributions on discrete spaces'' (Section~\ref{sec:prob}) in Appendices for those distributions).

Corollary~\ref{cor:WC_01} does not hold for $\gamma>1$, 
but in typical problems in bioinformatics the $\gamma$-centroid estimator for $\gamma>1$ can be calculated efficiently 
by using dynamic programming,
as shown in the following examples.

\begin{example}[$\gamma$-centroid estimator of pairwise alignment]\label{ex:align_gamma}
In Problem~\ref{prob:align} with Assumption~\ref{as:cl},
the $\gamma$-centroid estimator maximizes
the sum of the alignment probabilities
which are greater than ${1}/(\gamma+1)$ (Theorem~\ref{theorem:g_centroid_th}),
and for $\gamma \in [0,1]$
it can be given as the consensus estimator calculated from Eq.~(\ref{eq:gamma_consensus}) (Corollary~\ref{cor:WC_01}).
For $\gamma > 1$, the $\gamma$-centroid estimator is obtained
by using a dynamic programming algorithm with the same type of iterations
as in the Needleman-Wunsch algorithm:
\begin{align}
M_{i,k}&=\max\left\{
\begin{array}{l}
M_{i-1,k-1} + (\gamma+1)p_{ik}-1\\
M_{i-1,k}\\
M_{i, k-1}
\end{array}\right.\label{eq:DP_a_centroid}
\end{align}
where 
$M_{i,k}$ stores the optimal value of the alignment between two sub-sequences,
$x_{1}\cdots x_i$ and $x'_{1}\cdots x'_k$
(see ``Secondary structure prediction of an RNA sequence (Problem~\ref{prob:rnas})'' in Appendices for detailed descriptions). 
\end{example}

\begin{example}[$\gamma$-centroid estimator for prediction of secondary structures]\label{ex:sed_gamma}
In Problem~\ref{prob:rnas} with Assumption~\ref{as:cl},
the $\gamma$-centroid estimator maximizes
the sum of the base pairing probabilities
that are greater than ${1}/(\gamma+1)$ (Theorem~\ref{theorem:g_centroid_th}),
and for $\gamma \in [0,1]$
it can be given as the consensus estimator calculated from Eq.~(\ref{eq:gamma_consensus}) (Corollary~\ref{cor:WC_01}).
For $\gamma > 1$, the $\gamma$-centroid estimator is obtained
with the aid of a dynamic programming algorithm with the same type of iterations
as in the Nussinov algorithm:
\begin{align}
M_{i,j} = \max \left \{
\begin{array}{ll}
M_{i+1,j} \\
M_{i,j-1} \\
M_{i+1,j-1} + (\gamma+1) p_{ij} - 1\\
\max_k \left[M_{i,k} + M_{k+1,j}\right] 
\end{array}
\right.\label{eq:dp_centroid_rna_sec}
\end{align}
where $M_{i,j}$ stores the best
score of the sub-sequence $x_i x_{i+1}\cdots x_j$
(see ``Pairwise alignment of biological sequences (Problem~\ref{prob:align})'' in Appendices for the detail descriptions).
\end{example}
The $\gamma$-centroid estimators are implemented
in \last~\cite{Frith_BMCB} for Problem~\ref{prob:align} and
in \centroidfold~\cite{centroidfold-submit,pmid19435882} for Problem~\ref{prob:rnas}.
%

%
\begin{problem}[Estimation of phylogenetic trees]\label{prob:pt}
Given a set of operational taxonomic units $S$,
predict their phylogenetic trees (unrooted and multi-branched trees) as a point in
$\mathcal{T}(S)$, the space of all the possible phylogenetic trees of $S$.
\end{problem}

The phylogenetic tree in $\mathcal{T}(S)$ is represented as a binary vector with
$2^{n-1}-n-1$ dimension where $n$ is the number of units in $S$, 
based on partition of $S$ by cutting every edge in the tree 
(see ``The space of phylogenetic trees: $\mathcal{T}(S)$'' (Section~\ref{sec:Y_t}) in Appendices for details). 
A sampling algorithm can be used to estimate the partitioning probabilities 
approximately \cite{MetropolisEtAl:53}.

\begin{example}[$\gamma$-centroid estimator of phylogenetic estimation]\label{ex:tree_gamma}
In
Problem~\ref{prob:pt} with Assumption~\ref{as:cl}, the $\gamma$-centroid estimator
maximizes the number of the partitioning probabilities which are greater than
$1/(\gamma+1)$ (Theorem~\ref{theorem:g_centroid_th}), and for $\gamma\in[0,1]$ it can be give as the consensus
estimator calculated from Eq.~(\ref{eq:gamma_consensus}) (Corollary~\ref{cor:WC_01}) 
(see ``Estimation of phylogenetic trees (Problem~\ref{prob:pt})'' in Appendices for details).
\end{example}
Because the Hamming distance between two trees in $\mathcal{T}(S)$
is known as topological distance \cite{citeulike:515472}, the 1-centroid estimator minimizes the
expected topological distance.
In contrast to Example~\ref{ex:align_gamma} and Example~\ref{ex:sed_gamma}, it appears that
no method can efficiently compute the $\gamma$-centroid estimator with $\gamma>1$ in
Example~\ref{ex:tree_gamma}. Despite the difficulties of the application to phylogenetic trees, recently,
a method applying the concept of generalized centroid estimators was developed \cite{CWT}.
%
\subsection{Generalized centroid estimators for representative prediction}

%
Predictions based on probability distributions on the predictive space
were discussed in the previous sections.
However, there are certain even more complex problems in bioinformatics,
as illustrated by the following example.

\begin{problem}[Prediction of common secondary structures of RNA sequences]\label{prob:rep_rna}
Given a set of RNA sequences $D=\{x_i\}, i=1,\ldots K$ and their multiple alignment of length $L$
and the same energy model for each RNA sequence,
predict their common secondary structure
as a point in $\mathcal{S}'(L)$, which is the space of all possible secondary structures of length $L$.
\end{problem}

In the case of Problem~\ref{prob:rep_rna},
although the probability distribution is not implemented in the predictive space,
each RNA sequence $x_i$
has a probability distribution on its secondary structure derived from the energy model.
Therefore, the theories presented in the previous section cannot be applied directly to this problem.
However, if we devise a new type of gain function
that connects the predictive space
with the parameter space of the secondary structure of each RNA sequence,
we can calculate the expected gain over the distribution on the parameter spaces of RNA sequences.
In order to account for this type of problem in general,
we introduce Assumption~\ref{as:our} and Definition~\ref{def:general_gain} as follows.

\begin{assumption}\label{as:our}
In Problem~\ref{prob:1} 
there exists a probability distribution $p(\theta|D)$
on the parameter space $\Theta$
which might be different from the predictive space $Y$.
\end{assumption}

\begin{definition}[Generalized gain function]\label{def:general_gain}
In Problem~\ref{prob:1} with Assumption~\ref{as:our},
for a point $\theta \in \Theta$ and a prediction $y\in Y$,
a {\em generalized gain function} is defined as
$G:\Theta \times Y \to \R^+$, $G(\theta,y)$.
\end{definition}

It should be emphasized that the MEG estimator (Definition~\ref{def:MEG}), 
pointwise gain function (Definition~\ref{def:pgain_function}) and Theorem~\ref{theorem:general}
can be extended to the generalized gain function.

In the case of Problem~\ref{prob:rep_rna}, for example,
the parameter space is the product of the spaces of the secondary structures of each RNA sequence,
and the probability distribution is the product of the distributions of secondary structures of each RNA sequence.
Here, the general form of the problem of representative prediction is introduced.

\begin{problem}[Representative prediction]\label{prob:representative}
In Problem~\ref{prob:1} with Assumption~\ref{as:our},
if the parameter space is represented as a product space
($\Theta = \prod_{k=1}^{K}{\Theta^{(k)}} = Y^{K}$)
and the distribution of $\theta \in \Theta$ has the form
$p(\theta|D)=\prod_{k=1}^K p^{(k)}(\theta^k|D)$,
predict a point $y$ in the predictive space $Y$.
\end{problem}

The generalized gain function for the representative prediction should be chosen such
that the prediction
reflects as much as each data entry.
Therefore, it is natural to use the following generalized gain function
that integrates the gain for each parameter.

\begin{definition}[Homogeneous generalized gain function]\label{def:gain_additive}
In Problem~\ref{prob:representative},
a homogeneous generalized gain function is defined as
\begin{equation*}
G(\theta,y)=\sum_{k=1}^K G'(\theta^k,y),
\end{equation*}
where $G'$ is the gain function in Definition~\ref{def:GainFunction}.
\end{definition}

\begin{definition}[Representative estimator]\label{def:rep_est}
In Problem~\ref{prob:representative},
given a homogeneous generalized gain function
$G(\theta,y)=\sum_{k=1}^K G'(\theta^k,y)$,
the MEG estimator defined as
\begin{equation*}
\hat y^{(rMEG)} 
= \argmax_{y \in Y} \int G(\theta, y) p(\theta|D)d\theta
\end{equation*}
is referred to as the representative estimator.
\end{definition}

\begin{proposition}\label{prop:equiv_common_estimator}
The representative estimator is equivalent to
an MEG estimator with averaged probability
distribution on the predictive space $Y$:
\begin{align*}
p(y|D)=\frac{1}{K}\sum_{k} p^{(k)}(y|D)
\end{align*}
and a gain function $G'$.
\end{proposition}

This proposition shows
that a representative prediction problem with any homogeneous generalized gain function
can be solved in a manner similar to Problem~\ref{prob:1} ($\Theta = Y$) with averaged probability distribution.
Therefore, the $\gamma$-centroid estimator for a representative prediction
satisfies Corollary~\ref{cor:gamma_est_for_representative}.

\begin{corollary}\label{cor:gamma_est_for_representative}
In Problem~\ref{prob:representative}, 
the representative estimator
where $G'(\theta^k,y)$ is the gain function of the $\gamma$-centroid estimator on $Y$,
is
the $\gamma$-centroid estimator for the averaged probability distribution
and satisfies the same properties in
Theorem~\ref{thm:impl_g_est},
Theorem~\ref{theorem:g_centroid_th},
and Corollary~\ref{cor:WC_01}.
\end{corollary}

\subsection{Estimators based on marginal probabilities}\label{sec:approx_estimator}
%
In the previous section,
we introduced Assumption~\ref{as:our},
where there is a parameter space $\Theta$ that can be different from the predictive space $Y$,
and we discussed the problem of representative prediction.
In this section,
we discuss another type of problems where $\Theta \ne Y$.
An example is presented below.

\begin{problem}[Pairwise alignment using homologous sequences]\label{prob:align_homo}
Given a data set $D=\{x,x',h\}$,
where $x$ and $x'$ are two biological sequences to be aligned
and $h$ is a sequence that is homologous to both $x$ and $x'$,
predict a point $y$ in the predictive space $Y=\mathcal{A}(x,x')$
(the space of all possible alignments of $x$ and $x'$).
\end{problem}

The precise probabilistic model of this problem might include
the phylogenetic tree, ancestor sequences and their alignments.
Here, we assume a simpler situation
where the probability distribution of all possible multiple alignments of $D$ is given.
We predict the pairwise alignment of two specific sequences
according to the probability distribution of multiple alignments.
Although the parameter space $\Theta$, which is the space of all the possible multiple alignments,
can be parametrized using the parameters of the spaces of the alignments of all pairs that can be formed from the sequences in $D$,
$\Theta$ itself is not the product space of these spaces
because these pairwise alignments are not independent:
for $x ,x', h \in D$, $x_i$ must be aligned to $x'_j$
if both $x_i$ and $x'_j$ are aligned to $h_k$.
This type of problems can be generalized as follows.

\begin{problem}[Prediction in a subspace of the parameter space]\label{prob:pred_subspace}
In Problem~\ref{prob:1} with Assumption~\ref{as:our},
if the parameter space $\Theta$ is represented as
$\Theta \subset \Theta'\times\Theta'^\perp$,
predict a point $y$ in the predictive space $Y=\Theta'$. 
\end{problem}

For the problem of representative prediction (Problem~\ref{prob:representative}),
generalized gain functions on $\Theta \times Y$ were introduced
(Definition~\ref{def:general_gain} and Definition~\ref{def:gain_additive}).
In contrast, in Problem~\ref{prob:pred_subspace},
the values of the parameters in $\Theta'^\perp$ are not important, and a point in $Y=\Theta'$ is predicted.
In Problem ~\ref{prob:align_homo}, for example,
the optimal multiple alignment of $D$,
the pairwise alignment of $h$ and $x$, and
the pairwise alignment of $h$ and $x'$
are irrelevant,
but instead we predict the pairwise alignment of $x$ and $x'$.
The MEG estimator
for the gain function defined on $\Theta' \times Y$
can be written as
\begin{align*}
\hat y^{(sMEG)} 
&= \argmax_{y \in Y} \int G(\theta', y) p(\theta'|D)d\theta',
\end{align*}
where $p(\theta'|D)$ on $\Theta'$ is the marginalized distribution
\begin{align}
p(\theta'|D)
= \int p(\theta|D)d\theta'^\perp
= \int p(\theta',\theta'^\perp|D)d\theta'^\perp. \label{eq:marginalized_distribution}
\end{align}
From the above MEG estimator, it might appear that Problem~\ref{prob:pred_subspace} is trivial.
However, it is not a simple task
to calculate the marginalized distribution in Eq.~(\ref{eq:marginalized_distribution})
in actual problems.

To reduce the computational cost,
we change Problem~\ref{prob:pred_subspace}
by introducing an approximated probability distribution
on the product space $\Theta'\times\Theta'^\perp$ a follows.
\begin{problem}[Prediction in product space]\label{prob:pred_prod}
In Problem~\ref{prob:1} with Assumption~\ref{as:our},
if the parameter space $\Theta$ is represented as
$\Theta = \Theta'\times\Theta'^\perp$
and the probability distribution on $\Theta$ is defined as
\begin{align}
\bar{p}(\theta|D) = p(\theta'|D) p(\theta'^\perp|D), \label{Eq:13}
\end{align}
predict a point $y$ in the predictive space $Y=\Theta'$. 
\end{problem}

This factorization of spaces and probability distributions creates 
a number of inconsistencies in the parameter space with respect to the original Problem~\ref{prob:pred_subspace}.
In other words, the approximated distribution yields non-zero values
for a point that is not included in the original $\Theta$ (in Problem~\ref{prob:pred_subspace})
but in $\Theta'\times\Theta'^\perp$.
To reduce these inconsistencies,
a new type of gain function and a new estimator are introduced as follows.
%
\begin{definition}[$\gamma$-type pointwise gain function] \label{def:app_gamma_gain}
In Problem~\ref{prob:pred_subspace},
a $\gamma$-type pointwise gain function is defined as
$G(\theta,y)$ in Eq.~(\ref{eq:gain_function})
in Definition~\ref {def:pgain_function} having
\begin{align}
&F_i(\theta,y_i)=\gamma\cdot \delta_i(\theta') \cdot I(y_i=1) + (1-\delta_i(\theta')) I(y_i=0)\label{eq:pwg_approx},
\end{align}
where the value $\delta_i(\theta')\in[0,1]$ 
in the gain function should be designed to reduce the 
inconsistencies resulting from the factorization.
\end{definition}

\begin{definition}[Approximated $\gamma$-type estimator]\label{est:approx}
In Problem~\ref{prob:pred_prod},
with a $\gamma$-type pointwise gain function with $F_i(\theta,y_i)$ in Eq.~(\ref{eq:pwg_approx})
on $\Theta\times Y$,
an approximated $\gamma$-type estimator is defined as an MEG estimator:
\begin{align*}
\hat y^{({\gamma}app)} = \argmax_{y\in Y} \int \biggl[\sum_{i=1}^n F_i(\theta,y_i)\biggr] \bar{p}(\theta|D) d\theta.
\end{align*}
\end{definition}

\begin{example}[PCT in pairwise alignment]\label{ex:pct_alignment}
We obtain the approximate estimator 
for Problem~\ref{prob:align_homo} with the following settings.
The parameter space is given as $\Theta=\Theta'\times\Theta'^\perp$, where
\begin{align*}
\Theta'=\mathcal{A}(x,x') (=Y) \mbox{ and } \Theta'^\perp=\mathcal{A}(x,h)\times\mathcal{A}(x',h)
\end{align*}
and the probability distribution on the parameter space $\Theta$ is given as
\begin{align*}
p(\theta\left|D\right.)
= p^{(a)}(\theta^{xx'}|x,x')p^{(a)}(\theta^{xh}|x,h)
p^{(a)}(\theta^{x'h}|x',h)
\end{align*}
for 
$\theta=(\theta^{xx'}, \theta^{xh}, \theta^{x'h})\in\Theta=\Theta'\times\Theta'^\perp$.
The $\delta_i(\theta')$ in Eq.~(\ref{eq:pwg_approx}) of the $\gamma$-type pointwise gain function is defined as
\begin{align*}
\delta_{ik}(\theta')= \frac{1}{2}\biggl\{I(\theta^{xx'}_{ik}=1) + \sum_v I(\theta_{iv}^{xh}=1)I(\theta_{kv}^{x'h}=1)\biggl\}.
\end{align*}
The approximated $\gamma$-type estimator for this $\gamma$-type pointwise gain function
is {employed in a part of} 
{\em probabilistic consistency transformation} (PCT) \cite{pmid15687296},
which is an important step toward accurate multiple alignments. 
See ``Pairwise alignment using homologous sequences'' (Section~\ref{sec:alignment_hom}) in Appendices for precise descriptions.
\end{example}

It is easily seen that Theorem~\ref{theorem:g_centroid_th} applies to the approximated $\gamma$-type estimator 
if $p_i$ in Theorem~\ref{theorem:g_centroid_th}
is changed as follows:
\begin{align*}
p_i=\int \delta_i(\theta')p(\theta'|D)d\theta'.
\end{align*}
Moreover, to confirm whether approximated $\gamma$-type estimator contains the consensus estimator for the same gain function,
it is only necessary to check if
\begin{align}
(\gamma+1)\left(\delta_i(\theta')+\delta_j(\theta')\right)-2 \le 0,\label{eq:condition_ae}
\end{align}
instead of Eq.~(\ref{eq:zyouken}) in Theorem~\ref{theorem:general}.
(Note that Theorem~\ref{theorem:general} can be extended to the {\em generalized} (pointwise) 
gain function: see Theorem~\ref{theorem:general_suppl}.)

%
%
%
%

\section{Discussion}

\subsection{Properties of the $\gamma$-centroid estimator}

In this paper,
general criteria for designing estimators are given
by the maximum expected gain (MEG) estimator (Definition~\ref{def:MEG}).
The Bayesian ML estimator is an MEG estimator
with the delta function $\delta(y,\theta)$ as the gain function,
which means that only the probability for the ``perfect match'' is counted.
To overcome the drawbacks of the Bayesian ML estimator,
the centroid estimator \cite{centroid}
takes into account the overall ensemble of solutions and minimizes the expected Hamming loss.
Because the Hamming loss is not the standard evaluation measures for actual problems,
we have proposed an estimator of a more general type, the $\gamma$-centroid estimator (Definition~\ref{def:GCE}),
which includes the centroid estimator as a special case, $\gamma=1$.
The $\gamma$-centroid estimator is an MEG estimator that maximizes the expected value of $TN + \gamma TP$,
which generally covers all possible linear combination of the numbers of
true positives (TP), true negatives (TN), false positives (FP) and false negatives (FN) (Theorem~\ref{thm:impl_g_est}).
Since most of the evaluation measures of the prediction accuracy
are functions of these numbers \cite{pmid10871264},
the $\gamma$-centroid estimator is related to the principle of maximum expected accuracy (MEA).
It should be noted that MEG estimators have been proposed that are similar to the $\gamma$-centroid estimator for some specific problems,
for example, the alignment metric accuracy (AMA) estimator \cite{schwartz-2005} 
(see Section~\ref{sec:alignment} for the formal definition) 
for pairwise alignment (Problem~\ref{prob:align})
and
the MEA-based estimator \cite{pmid16873527} (see Appendices for the formal definition) 
for prediction of secondary structure of RNA (Problem~\ref{prob:rnas}).
However, these estimators display a {\em bias} with respect to the accuracy measures for the problem 
(see Eqs.~(\ref{eq:rel_ama}) and (\ref{eq:gf_mea})), and are therefore inappropriate from the viewpoint of the principles of MEA.
Moreover, these estimators cannot be introduced in a general setting, that is, Problem~\ref{prob:1}.
It has been also shown that the $\gamma$-centroid estimator
outperforms the MEA-based estimator~\cite{pmid16873527} for various probability distributions 
in computational experiments~\cite{centroidfold-submit}.
(See ``Pairwise alignment of biological sequences (Problem~\ref{prob:align})'' and
``Secondary structure prediction of an RNA sequence (Problem~\ref{prob:rnas})'' 
in
Appendices for relations between the $\gamma$-centroid estimator
and other estimators in Problems~\ref{prob:align} and \ref{prob:rnas}, respectively.)
%
\subsection{How to determine the parameter in $\gamma$-centroid estimator}
%
The parameter $\gamma$ in $\gamma$-centroid estimators adjusts 
sensitivity and PPV
(whose relation is tradeoff).
MCC or F-score is often used to obtain a balanced measure between sensitivity and PPV.
In RNA secondary structure predictions, it has been confirmed that 
the best $\gamma$ (with respect to MCC) 
of the $\gamma$-centroid estimator
with CONTRAfold model was larger than that with McCaskill model \cite{centroidfold-submit}.
It shows that the best $\gamma$ (with respect to a given accuracy measure) 
depends on not only estimation problems but also probabilistic models 
for predictive space. The parameter $\gamma$ trained by using reference structures
was therefore employed as the default parameter in CentroidFold \cite{centroidfold-submit}.
In order to select the parameter automatically (with respect to a given 
accuracy measure such as MCC and F-score), 
an approximation of maximizing expected MCC (or F-score) with the $\gamma$-centroid estimator 
can be utilized \cite{pmcc}.

%
\subsection{Accuracy measures and computational efficiency}
The reader might consider that it is possible to design estimators
that maximize the expected MCC or F-score
which balances sensitivity (SEN) and positive predictive value (PPV).
However,
it is much more difficult 
to compute such estimators in comparison with the $\gamma$-centroid estimator, as described below.

The expected value of the gain function of the $\gamma$-centroid estimator
can be written with marginalized probabilities as in Eq.~(\ref{eq:marginal_form}),
which can be efficiently computed by dynamic programming in many problems in bioinformatics,
for example, the forward-backward algorithm for alignment probabilities
and the McCaskill algorithm for base pairing probabilities.
Under a certain condition of the predictive space, which many problems in bioinformatics satisfy,
the $\gamma$-centroid estimator maximizes the sum of marginalized probabilities
greater than $1/(\gamma +1)$ (Theorem~\ref{theorem:g_centroid_th}).
Moreover, under an additional condition of the predictive space and the pointwise gain function,
which again many problems in bioinformatics satisfy,
the $\gamma$-centroid estimators for $\gamma \in [0,1]$ can be easily calculated
as the consensus estimators,
which collect in the binary predictive space
the components that have marginalized probabilities greater than $1/(\gamma+1)$ (Corollary~\ref{cor:WC_01}).
For $\gamma >1$, there often exist dynamic programming algorithms
that can efficiently compute the $\gamma$-centroid estimators (Examples~\ref{ex:sed_gamma} \& \ref{ex:align_gamma}),
but there are certain problems, such as Problem~\ref{prob:pt},
which seem to have no efficient dynamic programming algorithms.

The gain function of the estimators that maximize MCC or F-score, and also SEN or PPV
contain {\em multiplication} and/or {\em division} of TP, TN, FP and FN,
while the gain function of the $\gamma$-centroid 
estimator contains only the weighted {\em sums} of these values (i.e., $\mbox{TN}+\gamma\cdot\mbox{TP}$).
Therefore, the expected gain is not written with marginalized probabilities as in Eq.~(\ref{eq:marginal_form}),
and it is difficult to design efficient computational algorithms for those estimators.
In predicting secondary structures of RNA sequences (Problem~\ref{prob:rnas}), for example,
it is necessary to enumerate all candidate secondary structures
or sample secondary structures for an approximation
in order to compute the expected MCC/F-score of a predicted secondary structure. 
%

\subsection{Probability distributions are not always defined on predictive space}

After discussing the standard estimation problems on a binary space
where the probability distribution is defined on the predictive space,
we have proposed a new category of estimation problems
where the probability distribution is defined on a parameter space that differs from the predictive space (see Assumption~\ref{as:our}).
Two types of estimators for such problems,
for example, estimators for representative prediction
and estimators based on marginalized distribution,
have been discussed.

Prediction of the common secondary structure from an alignment of RNA sequences
(Problem~\ref{prob:rep_rna})
is an example of representative prediction.
The probability distribution is not implemented in the predictive space,
the space of common secondary structure,
but each RNA sequence has a probability distribution for its secondary structure.
Because the ``correct'' reference for the common secondary structure is not known in general,
direct evaluation of the estimated common secondary structure is difficult.
In the popular evaluation process for this problem,
the predicted common secondary structure is mapped to each RNA sequence
and compared to its reference structure. 
Using the homogeneous generalized gain function exactly implements this evaluation process
and the MEG estimator for the averaged probability distribution
is equivalent to the MEG estimator for homogeneous generalized gain function.
Therefore, we can use the averaged base pairing probabilities according to the alignment
as the distribution for the common secondary structure
(see ``Common secondary structure prediction from a multiple 
alignment of RNA sequences'' (Section~\ref{sec:rna_sec_pred_m}) in Appendices for detailed discussion).
The representative estimator for Problem~\ref{prob:rep_rna} is implemented in software
\centroidalifold.
Another example of representative prediction is the
``alignment of alignments'' problem,
which is the fundamental element of progressive multiple alignment of biological sequences.
The evaluation process using the sum of pairs score
corresponds to using the homogeneous generalized gain function.
(see ``Alignment between two {\em alignments} 
of biological sequences'' (Section~\ref{sec:alignment_align}) in Appendices for detailed discussion).

Estimation problems of marginalized distributions
can be formalized as prediction in a subspace of the parameter space (Problem~\ref{prob:pred_subspace}).
If we can calculate the marginalized distribution on the predictive space
from the distribution on the parameter space,
all general theories apply to the predictive space and the marginalized distribution.
In actual problems, such as pairwise alignment using homologous sequences (Problem~\ref{prob:align_homo}), however,
computational cost for calculation of the marginalized probability is quite high.
We introduced the factorized probability distribution (Eq.~(\ref{Eq:13})) for approximation,
the $\gamma$-type pointwise gain function (Definition~\ref{def:app_gamma_gain}) to reduce the inconsistency caused by the factorization,
and the approximated $\gamma$-type estimator (Definition~\ref{est:approx}).
In Problem~\ref{prob:align_homo},
the probability consistency transformation (PCT),
which is widely used for multiple sequence alignment,
is interpreted as an approximated $\gamma$-type estimator.
Prediction of secondary structures of RNA sequences
on the basis of homologous sequences \cite{pmid19478007} (see Problem~\ref{prob:rna_sec_pred_hom} in Appendices)
and pairwise alignment for {\em structured} RNA sequences
are further examples of this type of problems.

\subsection{Application of $\gamma$-centroid estimator to cluster centroid}
In case probability distribution on the predictive space is multi-modal, 
$\gamma$-centroid estimators can provide unreliable solutions.
For example, when there are two clusters of secondary structures 
in predictive spaces and those structures are exclusive,
the $\gamma$-centroid estimator might give
a ``chimeric'' secondary structure whose free energy is quite high. 
To avoid this situation, Ding {\it et al.} \cite{pmid16043502}
proposed a notion of the {\em cluster centroid}, which is
computed by the centroid estimator with a given cluster in a predictive space. 
We emphasize that the extension of cluster centroid by using $\gamma$-centroid estimator
is straightforward and would be useful. 


\subsection{Conclusion}

In this work, 
we constructed a general framework for designing 
estimators for estimation problems in high-dimensional discrete (binary) 
spaces.
The theory is regarded as a generalization of the pioneering work
conducted by Carvalho and Lawrence, and is closely related to the concept of MEA.
Furthermore, we presented several applications of the proposed estimators (see Table~\ref{tab:sumary} for summary) 
and the underlying theory.
The concept presented in this paper is highly extendable and sheds new light on many
problems in bioinformatics.
In future research, we plan to investigate further applications of 
the $\gamma$-centroid and related estimators presented in this paper.

\section*{Acknowledgments}

The authors are grateful to Drs.\ Luis E.\ Carvalho, 
Charles E.\ Lawrence, Kengo Sato, Toutai Mituyama and Martin C.\ Frith 
for fruitful discussions.
The authors also thank the members of the bioinformatics group for RNA
at the National Institute of Advanced Industrial Science and Technology (AIST)
for useful discussions.

{\small
{%
\begin{landscape}\label{sec:sum_app_bioinf}
%
%
\begin{table}[h]
\caption{\label{tab:sumary}
Summary of applications in bioinformatics
}
\begin{center}
{\small
\begin{tabularx}{225mm}{lXXXX}
\toprule
Alignment &
(1) Pairwise alignment of biological sequences &
(4) Pairwise alignment of two multiple alignments &
(6) Pairwise alignment using homologous sequences &
\\
\midrule
Section &
Section~\ref{sec:alignment} &
Section~\ref{sec:alignment_align} &
Section~\ref{sec:alignment_hom} &
\\
Data $D$ &
$\{x, x'\}$ & 
$\{A,A'\}$ & 
$\{x,x',H\}$ & 
\\
Predictive space $Y$ & 
$\mathcal{A}(x,x')$ &
$\mathcal{A}(A,A')$ & 
$\mathcal{A}(x,x')$ & 
\\
Parameter space $\Theta$ & 
$\mathcal{A}(x,x')$ &
$\prod_{x\in A}\prod_{x'\in A'}\mathcal{A}(x,x')$ &
$\mathcal{A}(x,x')\times\prod_{h\in H}[\mathcal{A}(x,h)\times\mathcal{A}(x',h)]$ &
\\
Probability $p(\theta|D)$ &
$p^{(a)}(\theta|x,x')$ &  
$\prod_{x\in A}\prod_{x'\in A'}p^{(a)}(\theta|x,x')$ &
$
p^{(a)}(\theta^{xx'}|x,x') \prod_{h\in H} 
[p^{(a)}(\theta^{xh}|x,h)
p^{(a)}(\theta^{x'h}|x',h)]$ &
\\
Type of estimator &
$\gamma$-centroid &
representative &
approximate &
\\
Software &
{\sc LAST} &
$-$ &
$-$ &
\\
Reference &
\cite{Frith_BMCB} &
\cite{pmid15687296}, This work &
\cite{pmid15687296}, This work 
\\
\midrule
RNA & 
(2) Secondary structure prediction of RNA  & 
(5) Common secondary structure prediction  &
(7) Secondary structure prediction using homologous sequences &
(8) Pairwise alignment of structured RNAs
\\
\midrule
Section &
Section~\ref{sec:rna_sec_pred} &
Section~\ref{sec:rna_sec_pred_m} &
Section~\ref{sec:rna_sec_pred_hom} &
Section~\ref{sec:rna_alignment} 
\\
Data $D$ & 
$\{x\}$ & 
$\{A\}$ &
$\{x,H\}$ & 
$\{x, x'\}$
\\
Predictive space $Y$ &
$\mathcal{S}(x)$ & 
$\mathcal{S}(A)$ & 
$\mathcal{S}(x)$ & 
$\mathcal{A}(x,x')$
\\
Parameter space $\Theta$ &
$\mathcal{S}(x)$ & 
$\prod_{x\in A} \mathcal{S}(x)$ & 
$\mathcal{S}(x)\times\prod_{h\in H}
\left[\mathcal{A}(x,h)\times\mathcal{S}(h)\right]$ &
$\mathcal{A}(x,x') \times \mathcal{S}(x) \times \mathcal{S}(x')$
\\
Probability $p(\theta|D)$ &
$p^{(s)}(\theta|x)$ &
$\prod_{x\in {A}}p^{(s)}(\theta|x)$ &
{$p^{(s)}\left(\theta^x|x\right) \times\prod_{h\in D} 
\left[p^{(a)}(\theta^{xh}|x,h)
p^{(s)}(\theta^{h}|h)\right]$} &
$p^{(a)}(\theta^{xx'}|x,x')p^{(s)}(\theta^{x}|x)p^{(s)}(\theta^{x'}|x')$
\\
Type of estimator &
$\gamma$-centroid &
representative &
approximate &
approximate
\\
Software &
{\sc CentroidFold} &
{\sc CentroidAlifold} &
{\sc CentroidHomfold} &
{\sc CentroidAlign}
\\
Reference &
\cite{centroidfold-submit} &
\cite{centroidfold-submit,CentroidAlifold} &
\cite{pmid19478007} &
\cite{pmid19808876}
\\
\midrule
Phylogenetic tree & (3) Estimation of phylogenetic tree \\
\midrule
Section & Section~\ref{sec:pt}
\\
Data $D$ & $S$
\\
Parameter space $\Theta$ & $\mathcal{T}(S)$
\\
Predictive space $Y$ & $\mathcal{T}(S)$
\\
Probability $p(\theta|D)$ & $p^{(t)}(\theta|S)$
\\
Type of estimator & $\gamma$-centroid
\\
Reference & This work
\\
\bottomrule
\end{tabularx}
}
\end{center}
{\footnotesize
The top row includes problems about RNA secondary structure predictions and
the {\color{black} middle} row includes problems about alignment of biological sequences.
Note that the estimators in the same column corresponds to each other.\\
}
\end{table}

\end{landscape}

}

\appendix

  \renewcommand{\thefigure}{S\arabic{figure}}
  \renewcommand{\thetable}{S\arabic{table}}
  \renewcommand{\theequation}{S\arabic{equation}}
  \setcounter{figure}{0}
  \setcounter{table}{0}
  \setcounter{equation}{0}
  \setcounter{footnote}{0}



\def\centroidfold{\textsc{CentroidFold}}
\def\centroidalifold{\textsc{CentroidAlifold}}
\def\centroidhomfold{\textsc{CentroidHomfold}}
\def\centroidalign{\textsc{CentroidAlign}}
\def\contrafold{\textsc{CONTRAfold}}
\def\rnaalifold{\textsc{RNAalifold}}
\def\rnaalifoldcentroid{\textsc{RNAalipfold-Centroid}}
\def\petfold{\textsc{PETfold}}
\def\mccaskill{\textsc{McCaskill}}
\def\mccaskillmea{\textsc{McCaskill-MEA}}
\def\pfold{\textsc{Pfold}}
\def\sfold{\textsc{Sfold}}
\def\rnaalipfold{\textsc{RNAalipfold}}
\def\probcons{\textsc{ProbCons}}
\def\mafft{\textsc{MAFFT}}
\def\clustalw{\textsc{ClustalW}}
\def\mxscarna{\textsc{MXSCARNA}}
\def\mfold{\textsc{Mfold}}
\def\rnafold{\textsc{RNAfold}}
\def\rnastructure{\textsc{RNAstructure}}
\def\simfold{\textsc{Simfold}}
\def\rfam{\textsc{Rfam}}
\def\last{\textsc{LAST}}
\def\blast{\textsc{BLAST}}
\def\rnastrand{\textsc{RNA STRAND}}

\newtheorem{remark}{Remark}
\newtheorem{property}{Property}
\newtheorem{lemma}{Lemma}
\newtheorem{proof}{(proof)}
\newtheorem{eproof}{(simple proof)}
\newtheorem{formulate}{Formulation}
\renewcommand{\theproof}{}
\renewcommand{\theeproof}{}

\def\enstate{\hfill\ensuremath{\Box}}
\def\enproof{\hfill\ensuremath{\blacksquare}}
\def\argmax{\mathop{\mathrm{arg\ max}}}
\def\R{\ensuremath{\mathbb{R}}}
\def\Z{\ensuremath{\mathbb{Z}}}
\def\N{\ensuremath{\mathbb{N}}}

{\color{black}
\section{Appendices}
}
%

%
\subsection{Discrete (binary) spaces in bioinformatics}\label{sec:discrete_space}
%
In this section, we summarize three discrete spaces that appear in this paper.
These discrete spaces are often used in the definition of the predictive spaces 
and the parameter spaces.
It should be noted that every discrete space described below is identical in form
to Eq.~(\ref{eq:constrain_Y}).
%
\subsubsection{The space of alignments of two biological sequences: $\mathcal{A}(x,x')$}\label{sec:Y_a}
%
We define a space of the alignments of two biological (DNA, RNA and protein) sequences
$x$ and $x'$, denoted by $\mathcal{A}(x,x')$, as follows.
We set $I^{(0)}=\left\{(i,k)|1\le i\le |x|, 1\le k\le |x'|\right\}$ as
a base index set, and 
a binary variable $\theta_{ik}$ for $(i,k)\in I^{(0)}$ is defined by
\begin{equation*}
\theta_{ik} = \left\{
\begin{array}{ll}
1 & \mbox{positions $i$ in $x$ and $k$ in $x'$ are aligned}\\
0 & \mbox{positions $i$ in $x$ and $k$ in $x'$ are not aligned}
\end{array}
\right..
\end{equation*}
Then $\mathcal{A}(x,x')$ is a subset of 
$B:=\left\{\theta=\left\{\theta_{ik}\right\}_{(i,k)\in I^{(0)}}\Big| \theta_{ik} \in \{0,1\}\right\}$
and is defined by
\begin{align*}
\mathcal{A}(x,x') = 
\bigcap_{I \in \mathcal{I}} C(I),\quad
C(I)=\biggl\{x' \in B \bigg| \sum_{(i,k)\in I} \theta_{ik} \le 1 \biggr\}.\label{eq:Y_alignment}
\end{align*}
Here $\mathcal{I}$ is a set of index-sets:
\begin{align*}
\mathcal{I} = \left\{ I \left| 
I=I_{i}^{(1)}\ (1\le i \le |x|) \mbox{ or }
I=I_{k}^{(2)}\ (1\le k \le |x'|) \mbox{ or }
I=I_{ikjl}^{(3)}\ (1 \le i<j \le {|x|}, 1\le l<k\le {|x'|})
\right.\right\}
\end{align*}
where
\begin{align*}
I_{i}^{(1)} = \left\{(i,k)|1\le k \le |x'|\right\},
I_{k}^{(2)} = \left\{(i,k)|1\le i \le |x|\right\} \mbox{ and }
I_{ikjl}^{(3)}=\left\{(i,k),(j,l)\right\}.
\end{align*}
The inclusion $y\in C({I_{i}^{(1)}})$ means that position $i$ in the sequence $x$ 
aligns with {\em at most one} position in the sequence $x'$ in the alignment $y$,
$y\in C({I_{j}^{(2)}})$ means that position $j$ in the sequence $x'$ 
aligns with {\em at most one} position in the sequence $x$ and
$y\in C({I_{ikjl}^{(3)}})$ means the alignment $(i,k)$ and $(j,l)$ 
is {\em not crossing}.
Note that $\mathcal{A}(x,x')$ depends on only the length of two sequences, namely, 
$|x|$ and $|x'|$.
%

\subsubsection{The space of secondary structures of RNA: $\mathcal{S}(x)$}\label{sec:Y_s}
%
We define a space of the secondary structures of an RNA sequence $x$, denoted by
$\mathcal{S}(x)$, as follows.
We set $I^{(0)}=\left\{(i,j)| 1\le i<j\le |x|\right\}$ as a
base index set, and 
a binary variable $\theta_{ij}$ for $(i,j)\in I^{(0)}$ is defined by
\begin{equation*}
\theta_{ij} = \left\{
\begin{array}{ll}
1 & \mbox{the positions $i$ of $x$ and $j$ of $x$ form a base pair}\\
0 & \mbox{the positions $i$ of $x$ and $j$ of $x$ do not form a base pair}.
\end{array}
\right.
\end{equation*}
Then $\mathcal{S}(x)$ is a subset of
$B:=\left\{\theta=\left\{\theta_{ij}\right\}_{(i,j)\in I^{(0)}}\Big| \theta_{ij} 
\in \{0,1\}\right\}$
and is defined by
\begin{equation*}
\mathcal{S}(x)= \bigcap_{I \in \mathcal{I}} C(I),\quad
C(I)=\biggl\{\theta \in B \bigg| \sum_{(i,j)\in I} \theta_{ij} \le 1 \biggr\}.
\end{equation*}
Here $\mathcal{I}$ is a set of index-sets 
\begin{align*}
\mathcal{I} = \left\{ I \left| I=I_i^{(1)}
\ (1\le i \le |x|) \mbox{ or } I=I_{ijkl}^{(2)}\ (1\le i<k<j<l\le |x|)\right.\right\}
\end{align*}
where
\begin{align*}
I_{i}^{(1)} = \left\{(i,j)| i<j\le|x|\right\}\cup\left\{(j,i)| 1\le j<i\right\} \mbox{ and }
I_{ijkl}^{(2)}=\left\{(i,j),(k,l)\right\}.
\end{align*}
The inclusion $y\in C({I_{i}^{(1)}})$ means that position $i$ in the sequence $x$ 
belongs to {\em at most one} base-pair in a secondary structure $y$, 
and 
$y\in C(I_{ijkl}^{(2)})$ means two base-pairs whose relation
is {\em pseudo-knot} are not allowed in $y$.
Note that $\mathcal{S}(x)$ depends on only the length of the RNA sequence $x$, that is, $|x|$.
%
%
\subsubsection{The space of phylogenetic trees: $\mathcal{T}(S)$}\label{sec:Y_t}
We define a space of phylogenetic trees (unrooted and multi-branch trees) 
of a set of $S=\{1,\cdots,n\}$, 
denoted by $\mathcal{T}(S)$, as follows.
We set $I^{(0)}=\left\{X|X\subset S^2, |X|<n/2 \vee (|X|=n/2 \wedge 1 \in X)\right\}$,
where
$S^2=\left\{X|X\subset S, |X|>1 \wedge |X|<n-1\right\}$, as a base index set and
we define binary variables $\theta_X$ for $X \in I^{(0)}$ by
\begin{equation*}
\theta_{X} = \left\{
\begin{array}{ll}
1 & \mbox{if $S$ can be partitioned into $X$ and $S\setminus X$ by cutting an edge in the tree}\\
0 & \mbox{otherwise}
\end{array}
\right..
\end{equation*}
Then $\mathcal{T}(S)$ is a subset of
$B:=\left\{\theta=\left\{\theta_{X}\right\}_{X\in I^{(0)}}\Big| \theta_X \in \{0,1\}\right\}$
and is defined by
\begin{equation*}
\mathcal{T}(S)= \bigcap_{I \in \mathcal{I}} C(I),\quad
C(I)=\biggl\{\theta \in B \bigg| \sum_{X\in I} \theta_{X} \le 1 \biggr\} 
\end{equation*}
where
$\mathcal{I} = \left\{ I=\{X,Y\} \left| X\cap Y\notin \left\{\emptyset,X,Y\right\} \right.\right\}$.
Note that $\mathcal{T}(S)$ depends on only the number of elements in $S$. 
We now give several properties of $\mathcal{T}(S)$ that follow directly from
the definition.
\begin{lemma}
The number of elements in $\mathcal{T}(S)$ (i.e. $|I^{(0)}|$) is equal to $2^{n-1}-n-1$ where $n=|S|$.
\end{lemma}
\begin{lemma}
The {\em topological distance} \cite{citeulike:515472} between two phylogenetic trees $T_1$ and $T_2$ 
in $\mathcal{T}(S)$ 
is 
\begin{align*}
d(T_1,T_2)=\sum_{X\in I^{(0)}} I(\theta_X(T_1) \ne \theta_X(T_2))
\end{align*}
where $I(\cdot)$ is the indicator function.
\end{lemma}
\begin{remark}
If we assume the additional condition 
$\sum_X \theta_X = ((4n-6)-2n)/2=n-3$, then $\mathcal{T}(S)$ is a set of {\em binary} trees.
\end{remark}

\subsection{Probability distributions on discrete spaces}\label{sec:prob}
%
We use three probability distributions in this paper.
%
\subsubsection{Probability distributions $p^{(a)}(\theta|x,x')$ on $\mathcal{A}(x,x')$}\label{sec:P_a}
For two protein sequences $x$ and $x'$, a probability distribution $p^{(a)}(\theta|x,x')$ 
over the space $\mathcal{A}(x,x')$, which is the space of pairwise alignments of
$x$ and $x'$ defined in the previous section, is given by the following models.
\begin{enumerate}
\item Miyazawa model~\cite{pmid8771180} and Probalign model~\cite{pmid16954142}:
\begin{align*}
p^{(a)}(\theta|x,x')=\frac{1}{Z(T)} \exp \left(\frac{S(\theta)}{T}\right)
\end{align*}
where $S(\theta)$ is the score of an alignment $\theta$
under the given scoring matrix 
(We define $S(\theta)=\sum_{\theta_{ij}=1} s(x_i,x_j) - \mbox{(penalty for gaps)}$
where $s(x_i,x_j)$ is a score for the correspondence of bases $x_i$
and $x_j$),
$T$ is the thermodynamic temperature
and $Z(T)$ is the normalization constant, which is known as a {\em partition function}.
\item Pair Hidden Markov Model (pair HMM)~\cite{pmid15687296}:
\begin{align*}
p^{(a)}(\theta|x,x') = \pi(s_1)\biggl(\prod_{i=1}^{n-1} \alpha(s_i\to s_{i+1})\biggr)
\biggl(\prod_{i=1}^n\beta(o_i|s_i)\biggr)
\end{align*}
where $\pi(s)$ is the initial probability of starting in state $s$,
$\alpha(s_i\to s_{i+1})$ is the transition probability from $s_i$ to $s_{i+1}$
and $\beta(o_i|s_i)$ is the omission probability for either a single letter or aligner residue 
pair $o_i$ in the state $s_i$.
\item CONTRAlign (pair CRF) model~\cite{DBLP:conf/recomb/DoGB06}:
\begin{align*}
p^{(a)}(\theta|x,x')=\frac{\exp(w^t f(\theta,x,x'))}{\sum_{\theta'\in \Omega(x,x')}\exp(w^t f(\theta',x,x'))}
\end{align*}
where $w$ is a parameter vector and 
$f(\theta, x,x')$ is a vector of features that indicates the number of times
each parameter appears, 
$\Omega(x,x')$ denotes the set of all possible alignments of $x$ and $x'$.
We do not describe the feature vectors and refer readers to 
the original paper~\cite{DBLP:conf/recomb/DoGB06}.
\end{enumerate}
\begin{remark}
Strictly speaking, the alignment space in the pair hidden Markov model 
and the CONTRAlign model consider the patterns of gaps.
In these cases, we obtain the probability space on $\mathcal{A}(x,x')$ 
by a {marginalization}.
\end{remark}
%
\subsubsection{Probability distributions $p^{(s)}(\theta|x)$ on $\mathcal{S}(x)$}\label{sec:P_s}
For an RNA sequence $x$, a probability distribution $p^{(s)}(\theta|x)$ 
over $\mathcal{S}(x)$, 
which is the space of secondary structures of $x$ defined in the previous section 
is given by the following models.
\begin{enumerate}
\item McCaskill model~\cite{pmid1695107}: This model is  based on the energy models for secondary structures
of RNA sequences  and is defined by
\begin{align*}
p^{(s)}(\theta|x)= \frac{1}{Z(x)} \exp \left(-\frac{E(\theta,x)}{kT}\right) \mbox{ where }
Z(x) = \sum_{\theta\in\mathcal{S}(x)} \exp \left(-\frac{E(\theta,x)}{kT}\right)
\end{align*}  
where $E(\theta,x)$ denotes the energy of the secondary structure that is 
computed using the energy parameters of Turner Lab~\cite{pmid10329189}, $k$ and $T$ are constants and
$Z(x)$ is the normalization term known as the {\em partition function}.
\item Stochastic Context free grammars (SCFGs) model~\cite{pmid15180907}:
\begin{align*}
p^{(s)}(\theta|x)=\frac{\sum_{\sigma\in\Omega(\theta)} p(x,\sigma)}{\sum_{\sigma\in\Omega'(x)} p(x,\sigma)}
\end{align*}
where $p(x,\sigma)$ is the joint probability of generating the parse $\sigma$
and is given by the product of the transition and emission probabilities of
the SCFG model and
$\Omega'(x)$ is all parses of $x$, 
$\Omega(\theta)$ is all parses for a given $\theta$.
\item CONTRAfold (CRFs; conditional random fields) model~\cite{pmid16873527}:
This model gives us the best performance on secondary structure prediction
although it is not based on the energy model. 
\begin{align*}
p^{(s)}(\theta|x) = \frac {\sum_{\sigma \in \Omega(\theta)} \exp (w^t f(x, \sigma))}
{\sum_{\sigma \in \Omega'(x)} \exp (w^t f(x, \sigma))} 
\end{align*}
where $w \in \R^n$, $f(x, \sigma) \in \R^n$ is the feature vector for $x$ in parse  $\sigma$, $\Omega'(x)$ is all parses of $x$, 
$\Omega(\theta)$ is all parses for a given $\theta$.
\end{enumerate}
%

\subsubsection{Probability distributions $p^{(t)}(\theta|S)$ on $\mathcal{T}(S)$}\label{sec:P_t}
%
A probability distribution $p^{(t)}(\theta|S)$ on $\mathcal{T}(S)$ is given 
by probabilistic models of phylogenetic trees, for example,~\cite{pmid12912839,pmid11524383}.
Those models give a probability distribution on binary trees and we should marginalize
these distributions for multi-branch trees.
%
\subsection{Evaluation measures defined using TP, TN, FP and FN}\label{sec:acc_mes}
%
There are several evaluation measures of a prediction in estimation problems 
for which we have a reference (correct) prediction
in Problem~\ref{prob:1}.
The Sensitivity (SEN), Positive Predictive Value (PPV),
Matthew's correlation coefficient (MCC) and F-score
for a prediction are defined as follows.
\begin{align*}
&\mbox{SEN} = \frac{\mbox{TP}}{\mbox{TP}+\mbox{FN}},\\
&\mbox{PPV} = \frac{\mbox{TP}}{\mbox{TP}+\mbox{FP}},\\
&\mbox{MCC}=\frac{\mbox{TP}\times \mbox{TN}-\mbox{FP}\times \mbox{FN}}
{\sqrt{(\mbox{TP}+\mbox{FP})(\mbox{TP}+\mbox{FN})(\mbox{TN}+\mbox{FP})(\mbox{TN}+\mbox{FN})}},\\
&\mbox{F-score}=\frac{2\cdot\mbox{TP}}{2\cdot\mbox{TP}+\mbox{FP}+\mbox{FN}}
\end{align*}
where TP (the number of true positive), TN (the number of true negative), 
FP (the number of false positive) and FN (the number of false negative) 
are defined by 
\begin{align}
&\mbox{TP}=\mbox{TP}(\theta,y)=\sum_{i}I(y_{i}=1)I(\theta_{i}=1)\label{eq:tp},\\
&\mbox{TN}=\mbox{TN}(\theta,y)=\sum_{i}I(y_{i}=0)I(\theta_{i}=0)\label{eq:tn},\\
&\mbox{FP}=\mbox{FP}(\theta,y)=\sum_{i}I(y_{i}=1)I(\theta_{i}=0)\label{eq:fp},\\
&\mbox{FN}=\mbox{FN}(\theta,y)=\sum_{i}I(y_{i}=0)I(\theta_{i}=1).\label{eq:fn}
\end{align}
It should be noted that these measures can be written as a function of TP, TN, FP and FN.
See \cite{pmid10871264} for other evaluation measures.
%
\subsection{Schematic diagrams of representative and approximated $\gamma$-type estimators}\label{sec:figs}
%
The schematic diagrams of the MEG estimator (Definition~\ref{def:MEG}{}),
the representative estimator (Definition~\ref{def:rep_est}{})
and the approximated $\gamma$-type estimator (Definition~\ref{est:approx}{})
are shown in Figure~\ref{fig:meg_est}, Figure~\ref{fig:rep_est} and Figure~\ref{fig:approx_est}, respectively.
\begin{figure}[th]
\centerline{
\includegraphics[width=0.6\textwidth]{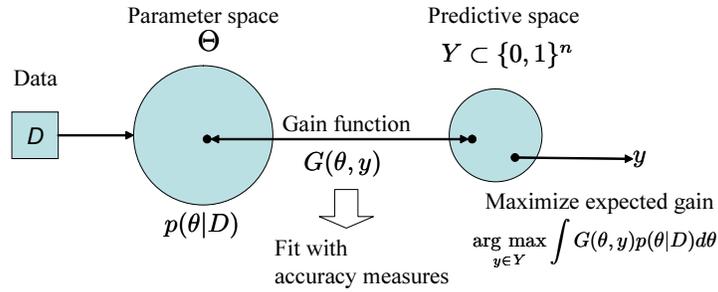} 
}
\caption{\label{fig:meg_est}
{\bf Schematic diagram of the MEG estimator (Definition~\ref{def:MEG}).}
}
\end{figure}
\begin{figure}[th]
\centerline{
\includegraphics[width=0.7\textwidth]{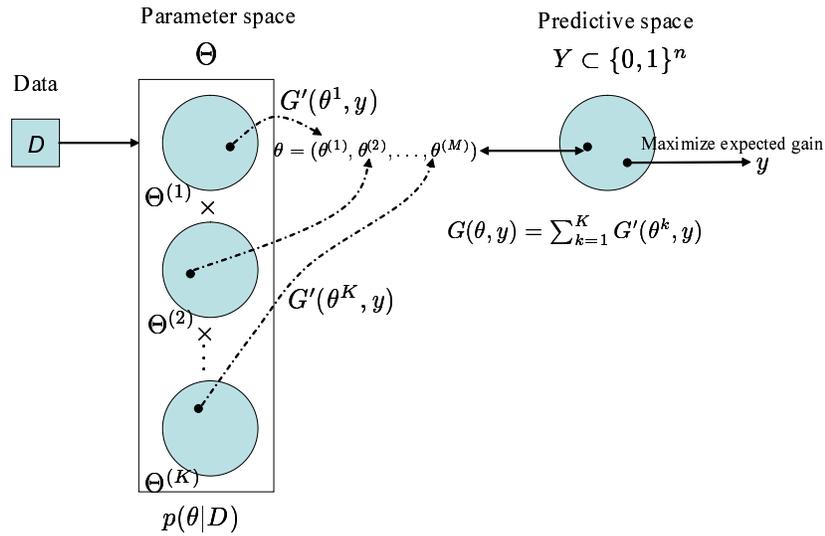} 
}
\caption{\label{fig:rep_est}
{\bf Schematic diagram of the representative estimator (Definition~\ref{def:rep_est}).}
The parameter space $\Theta$ is a product space and is different from the predictive space $Y$.
}
\end{figure}
\begin{figure}[th]
\centerline{
\includegraphics[width=0.65\textwidth]{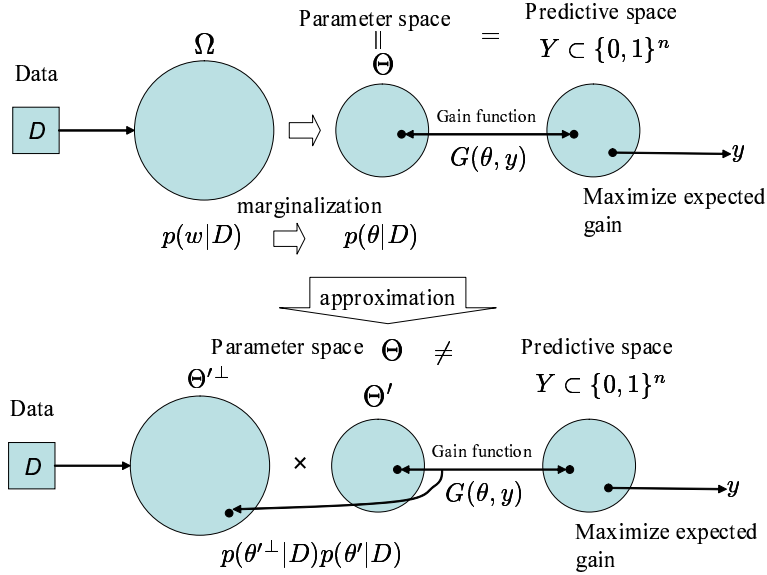} 
}
\caption{\label{fig:approx_est}
{\bf Schematic diagram of the approximated $\gamma$-type estimator (Definition~\ref{est:approx}).}
The estimator in the top figure shows the $\gamma$-centroid estimator with the
marginalized probability distribution, and the one in the bottom figure shows
its approximation.
}
\end{figure}
%
%
\subsection{Applications in bioinformatics}\label{sec:app}
%
In this section we describe several applications to bioinformatics of the general theories.
Some of these applications have already been published.
In those cases, we briefly explain the applications and the readers should see the
original paper for further descriptions as well as the computational experiments.
All of the applications in this section are summarized 
in Table~\ref{tab:sumary}.
%

\subsubsection{Pairwise alignment of biological sequences (Problem~\ref{prob:align})}\label{sec:alignment}

The pairwise alignment of biological (DNA, RNA, protein) sequences 
(Problem~\ref{prob:align}) is another fundamental and important problem 
of sequence analysis in bioinformatics (cf.~\cite{Durbin:1998}).


The $\gamma$-centroid estimator for Problem~\ref{prob:align} can be introduced as follows:

\begin{estimator}[$\gamma$-centroid estimator for Problem~\ref{prob:align}]\label{est:align}
For Problem~\ref{prob:align}, we obtain the $\gamma$-centroid estimator where
the predictive space $Y$ is equal to $\mathcal{A}(x,x')$  and 
the probability distribution on $Y$ is taken by 
$p^{(a)}(\theta|x,x')$.
\end{estimator}
First, Theorem~\ref{thm:impl_g_est}
and the definition of $\mathcal{A}(x,x')$ lead to the following property.
\begin{property}[A relation of Estimator~\ref{est:align} with accuracy measures]
The $\gamma$-centroid estimator for Problem~\ref{prob:align}{} is suitable for the accuracy measures: 
SEN, PPV, MCC and F-score {\em with respect to the aligned-bases} in the predicted alignment.
\end{property}
Note that accurate prediction of aligned-bases is important for the analysis of 
alignments, for example, in phylogenetic analysis.
Therefore, the measures in above are often used in evaluations of alignments e.g. \cite{Frith_BMCB}.

The marginalized probability 
$p_{ik}=p^{(a)}(\theta_{ik}=1|x,x')=\sum_{\theta\in\mathcal{A}(x,x')}I(\theta_{ik}=1)p^{(a)}(\theta|x,x')$
is called the {\em aligned-base (matching) probability} in this paper.
The aligned-base probability matrix $\{p_{ik}\}_{i,k}$ can be computed 
by the forward-backward algorithm whose time complexity is equal to $O(|x||x'|)$~\cite{Durbin:1998}.
Now, Theorem~\ref{theorem:g_centroid_th} leads to the following property.
\begin{property}[Computation of Estimator~\ref{est:align}]
The pairwise alignment of Estimator~\ref{est:align} is found by maximizing
the sum of aligned-base probabilities $p_{ik}$ 
(of the aligned-bases in the predicted alignment) that are larger
than $1/(\gamma+1)$. 
Therefore, it can be 
computed by a Needleman-Wunsch-style dynamic programming (DP) algorithm~\cite{pmid5420325} after calculating
the aligned-base matrix $\{p_{ik}\}$:
\begin{align}
M_{i,k}&=\max\left\{
\begin{array}{l}
M_{i-1,k-1} + (\gamma+1)p_{ik}-1\\
M_{i-1,k}\\
M_{i, k-1}
\end{array}\right.\label{eq:DP_a_centroid_suppl}
\end{align}
where 
$M_{i,k}$ stores the optimal value of the alignment between two sub-sequences,
$x_{1}\cdots x_i$ and $x'_{1}\cdots x_k$.
\end{property}
%
The time complexity of the recursion of the DP algorithm in Eq.~(\ref{eq:DP_a_centroid_suppl}) is 
equal to $O(|x||x'|)$, so the total computational
cost for predicting the secondary structure of the $\gamma$-centroid estimator remains
$O(|x||x'|)$.

By using Corollary~\ref{cor:WC_01}, we can predict the pairwise alignment 
of Estimator~\ref{est:align} with $\gamma\in[0,1]$ without using the DP algorithm in Eq.~(\ref{eq:DP_a_centroid_suppl}).
\begin{property}[Computation of Estimator~\ref{est:align} with $0\le\gamma\le 1$]
The pairwise alignment of the $\gamma$-centroid estimator can be predicted 
by collecting the aligned-bases whose probabilities are larger than $1/(\gamma+1)$.
\end{property}
The genome alignment software called \last~({\tt http://last.cbrc.jp/})~\cite{Frith_BMCB,pmid20110255} 
employs the $\gamma$-centroid estimator accelerated by an X-drop algorithm, and
the authors indicated that Estimator~\ref{est:align} reduced the false-positive
aligned-bases, compared to the conventional alignment (maximum score estimator).


Relations of Estimator~\ref{est:align} with existing estimators are summarized as follows:

\begin{enumerate}
\item A relation with the estimator by Miyazawa \cite{pmid8771180} (i.e. the centroid estimator):

Estimator~\ref{est:align} where $\gamma=1$ and the Miyazawa model 
is equivalent to the centroid estimator
proposed by Miyazawa \cite{pmid8771180}.
\item A relation with the estimator by Holmes {\it et al.}~\cite{pmid9773345}:

Estimator~\ref{est:align} with sufficiently large $\gamma$ is equivalent to the
estimator proposed by Holmes {\em et al.}, which 
maximizes the sum of matching probabilities in the predicted alignment.
\item A relation with the estimator in \probcons:
In the program, \probcons, Estimator~\ref{est:align} with pair HMM model and the sufficient large $\gamma$
was used. 
This means that \probcons{} only take care the sensitivity (or SPS) for the predicted alignment.
\item {A relation with the estimator by Schwartz {\it et al.}}:

For Problem~\ref{prob:align}, Schwartz {\it et al.}~\cite{schwartz-2005} proposed an 
Alignment Metric Accuracy (AMA) estimator, which is similar to the 
$\gamma$-centroid estimator (see also \cite{pmid18796475}).
The AMA estimator is a maximum gain estimator (Definition~\ref{def:MEG}) with the following gain function.
%
\begin{align*}
&G^{(\mathrm{AMA})}(\theta,y) = \\
&2 \sum_{i,j} I(\theta_{ij}=1)I(y_{ij}=1)
+ G_f \biggl\{ 
\sum_{i} \prod_{j} I(\theta_{ij}=0) I(y_{ij}=0)
+
\sum_{j} \prod_{i} I(\theta_{ij}=0) I(y_{ij}=0)
\biggr\}
\end{align*}
for $\theta, y\in\mathcal{A}(x,x')$.
In the above equation, $G_f\ge 0$ is a gap factor, which is a weight 
for the prediction of gaps.
We refer to the function $G^{(\mathrm{AMA})}(\theta,y)$ as the gain function of the AMA estimator.
In a similar way to that described in the previous section, we obtain a relation between 
$G^{(\mathrm{AMA})}(\theta,y)$ and $G^{(\mathrm{centroid})}(\theta,y)$ (the gain function of
the $\gamma$-centroid estimator). If we set $1/G_f=\gamma$,
then we obtain
\begin{align}
G^{(\mathrm{AMA})}(\theta,y) &= \frac{2}{\gamma}G^{(\mathrm{centroid})}(\theta,y) + \frac{1}{\gamma}
A(\theta,y)+ C_\theta\label{eq:rel_ama}
\end{align}
where 
\begin{align}
A(\theta,y)=
\sum_{i}\sum_{(j_1,j_2):j_1\ne j_2} I(\theta_{ij_1}=1)I(y_{ij_2}=1)
+\sum_{j}\sum_{(i_1,i_2):i_1\ne i_2} I(\theta_{i_1j}=1)I(y_{i_2j}=1)\nonumber
\end{align}
and $C_\theta$ is a value which does not depend on $y$.
If $I(\theta_{ij_1}=1)I(y_{ij_2}=1)=1$ for $j_1 \ne j_2$, then
we obtain $I(\theta_{ij_1}=1)I(y_{ij_1}=0)=1$ and
$I(\theta_{ij_2}=0)I(y_{ij_2}=1)=1$, and this means
that $(i,j_1)$ is an aligned pair that is a false negative and $(i,j_2)$ is an aligned pair that is a false positive
when $\theta$ is a reference alignment and $y$ is a predicted alignment.
Therefore, the terms $A(\theta,y)$ (in Eq.~(\ref{eq:rel_ama})) in the gain function of AMA are 
not appropriate for the evaluation measures SEN, PPV, MCC and F-score for aligned bases. 
In summary, the $\gamma$-centroid estimator is suitable for the evaluation measures:
SEN, PPV and F-score with respect to the aligned-bases while the AMA estimator is suitable for the AMA.  
\end{enumerate}

\subsubsection{Secondary structure prediction of an RNA sequence (Problem~\ref{prob:rnas})}\label{sec:rna_sec_pred}
%
Secondary structure prediction of an RNA sequence (Problem~\ref{prob:rnas}) 
is one of the most important problems of sequence analysis 
in bioinformatics.
Its importance has increased due to the recent discovery of 
functional non-coding RNAs (ncRNAs) because the
functions of ncRNAs are closely related to their secondary structures~\cite{pmid15608160}.
%
%

$\gamma$-centroid estimator for Problem~\ref{prob:rnas} can be introduced as follows:
\begin{estimator}[$\gamma$-centroid estimator for Problem~\ref{prob:rnas}]\label{est:rna_sec_pred}
For Problem~\ref{prob:rnas}, we obtain the $\gamma$-centroid estimator (Definition~\ref{def:GCE}) where
the predictive space $Y$ is equal to $\mathcal{S}(x)$  
and the probability distribution 
on $Y$ is taken by $p^{(s)}(\theta|x)$.
\end{estimator}
The general theory of the $\gamma$-centroid estimator 
leads to several properties.
First, the following property is derived from Theorem~\ref{thm:impl_g_est} 
and the definition of $\mathcal{S}(x)$.
\begin{property}[A relation of Estimator~\ref{est:rna_sec_pred} with accuracy measures]
The $\gamma$-centroid estimator for Problem~\ref{prob:rnas}{} is suitable for 
the widely-used accuracy measures of the RNA secondary structure prediction: 
SEN, PPV and MCC {\em with respect to base-pairs}
in the predicted secondary structure.
\end{property}
Because the base-pairs in a secondary structure are biologically important,
SEN, PPV and MCC with respect to base-pairs are widely used in evaluations 
of RNA secondary structure prediction,
for example, \cite{pmid16873527,centroidfold-submit,pmid17646296}.

The marginalized probability 
$p_{ij}=p^{(s)}(\theta_{ij}=1|x)=\sum_{\theta\in\mathcal{S}(x)}I(\theta_{ij}=1)p^{(s)}(\theta|x)$
is called a {\em base-pairing probability}.
The base-paring probability matrix $\{p_{ij}\}_{i<j}$ can be computed by the Inside-Outside
algorithm whose time complexity is equal to $O(|x|^3)$ where $|x|$ is the length of RNA sequence 
$x$~\cite{pmid1695107,Durbin:1998}.
Then, Theorem~\ref{theorem:g_centroid_th} leads to the following property.
\begin{property}[Computation of Estimator~\ref{est:rna_sec_pred}]
The secondary structure of Estimator~\ref{est:rna_sec_pred} is found by
maximizing the sum of the base-pairing probabilities $p_{ij}$ 
(of the base-pairs in the predicted structure)
that are larger than $1/(\gamma+1)$. 
Therefore, it can be 
computed by a Nussinov-style dynamic programming (DP) algorithm~\cite{Nussinov1978} after
calculating the base-pairing
probability matrix $\{p_{ij}\}$:
\begin{align}
M_{i,j} = \max \left \{
\begin{array}{ll}
M_{i+1,j} \\
M_{i,j-1} \\
M_{i+1,j-1} + (\gamma+1) p_{ij} - 1\\
\max_k \left[M_{i,k} + M_{k+1,j}\right] 
\end{array}
\right.\label{eq:dp_centroid_rna_sec_suppl}
\end{align}
where $M_{i,j}$ stores the best
score of the sub-sequence $x_i x_{i+1}\cdots x_j$.
\end{property}
If we replace ``$(\gamma+1) p_{ij} - 1$'' with ``$1$'' in Eq.~(\ref{eq:dp_centroid_rna_sec_suppl}), 
the DP algorithm is equivalent to the
Nussinov algorithm \cite{Nussinov1978} that maximizes the number of base-pairs in a predicted secondary
structure.
The time complexity of the recursion of the DP algorithm in Eq.~(\ref{eq:dp_centroid_rna_sec_suppl}) 
is equal to $O(|x|^3)$.
Hence, the total computational
cost for predicting the secondary structure of the $\gamma$-centroid estimator remains
$O(|x|^3)$, which is the same time complexity as for standard software: 
\mfold~\cite{pmid12824337}, \rnafold~\cite{hofacker1994} 
and \rnastructure~\cite{pmid15123812}.

By using Corollary~\ref{cor:WC_01}, we can predict the secondary structure of Estimator~\ref{est:rna_sec_pred}
with $\gamma\in[0,1]$ without using the DP algorithm in Eq.~(\ref{eq:dp_centroid_rna_sec_suppl}).
\begin{property}[Computation of Estimator~\ref{est:rna_sec_pred} with $0<\gamma\le 1$]
The secondary structure of the $\gamma$-centroid  estimator with $\gamma\in[0,1]$ can be predicted 
by collecting the base-pairs whose probabilities are larger than $1/(\gamma+1)$.
\end{property}
The software \centroidfold~\cite{centroidfold-submit,pmid19435882} implements 
Estimator~\ref{est:rna_sec_pred} with various probability distributions for the
secondary structures, such as the \contrafold{} and \mccaskill{} models.

%

Relations of Estimator~\ref{est:rna_sec_pred} with other estimators are summarized as follows:

\begin{enumerate}
\item A relation with the estimator used in \sfold~\cite{pmid16109749,pmid15215366}:

Estimator~\ref{est:rna_sec_pred} with $\gamma=1$ and the McCaskill model 
(i.e. the centroid estimator with the McCaskill model) is 
equivalent to the estimator used in the \sfold{} program.
\item A relation with the estimator used in \contrafold: 

For Problem~\ref{prob:rnas}, Do {\it et al.}~\cite{pmid16873527} proposed an MEA-based estimator,
which is similar to the $\gamma$-centroid estimator.
(The MEA-based estimator was also used in a recent paper~\cite{pmid19703939}.)
The MEA-based estimator is defined by the maximum expected gain estimator (Definition~\ref{def:MEG})
with the following gain function for $\theta$ and $y\in \mathcal{S}(x)$.
\begin{align}
&G^{(\mathrm{contra})}(\theta, y)=
\sum_{i=1}^{|x|} 
\Bigl[ 
{\gamma}\sum_{j:j\ne i}  I(\theta_{ij}^*=1)I(y_{ij}^*=1)+
\prod_{j:j\ne i} I(\theta_{ij}^*=0)I(y_{ij}^*=0)
\Bigl]\label{eq:gf_mea}
\end{align}
where $\theta^*$ and $y^*$ are symmetric extensions of (upper triangular matrices) 
$\theta$ and $y$, respectively
(i.e. $\theta^*_{ij}=\theta_{ij}$ for $i<j$ and $\theta^*_{ij}=\theta_{ji}$ for $j<i$;
the definition of $y^*$ is similar.).
It should be noted that, under the general estimation problem of Problem~\ref{prob:1},
the gain function of Eq.~(\ref{eq:gf_mea}) cannot be introduced, and the gain function 
is specialized for the problem of RNA secondary structure prediction.

The relation between the gain function of the $\gamma$-centroid estimator 
(denoted by $G^{(\mathrm{centroid})}(\theta,y)$ and defined in Definition~\ref{def:GCE}) 
and the one of the MEA-based estimator is
\begin{align}
G^{(\mathrm{contra})}(\theta,y)=G^{(\mathrm{centroid})}(\theta,y)+A(\theta,y)+C(\theta)
\end{align}
where the additional term $A(\theta,y)$ is positive for {\em false} predictions
of base-pairs (i.e., $\mbox{FP}$ and $\mbox{FN}$) and $C(\theta)$ does not depend on 
the prediction $y$ (see~\cite{centroidfold-submit} for the proof).
This means the MEA-based estimator by Do et al. possess a bias against the widely-used accuracy measures for Problem~\ref{prob:rnas}
(SEN, PPV and MCC of base-pairs) 
compared with the $\gamma$-centroid estimator.
Thus, the $\gamma$-centroid estimator is theoretically 
superior to the MEA-based estimator by Do et al. with respect to those accuracy measures.
In computational experiments, the authors confirmed that the $\gamma$-centroid
estimator is always better than the MEA-based estimator when we used the same probability
distribution of secondary structures.
See~\cite{centroidfold-submit} for details of the computational experiments.
\end{enumerate}

\subsubsection{Estimation of phylogenetic trees (Problem~\ref{prob:pt})}\label{sec:pt}




The $\gamma$-centroid estimator for Problem~\ref{prob:pt} can be introduced as follows:

\begin{estimator}[$\gamma$-centroid estimator for Problem~\ref{prob:pt}]\label{est:pt}
For Problem~\ref{prob:pt}, we obtain the $\gamma$-centroid estimator (Definition~\ref{def:GCE})
where the predictive space $Y$ is equal to $\mathcal{T}(S)$
and the probability distribution on $Y$ is taken by $p^{(t)}(\theta|S)$.
\end{estimator}

The following property is easily obtained by Theorem~\ref{thm:impl_g_est} and \cite{citeulike:515472}.
\begin{property}[Relation of 1-centroid estimator and topological distance]
The $\gamma$-centroid estimator with $\gamma=1$ (i.e. centroid estimator) 
for Problem~\ref{prob:pt}{} minimizes expected topological distances.
\end{property}

For $X\in I^{(0)}$ ($I^{(0)}$ is a set of partitions of $S$ and is formally defined in the previous section), 
we call the marginalized probability $p_X = \sum_{\theta\in\mathcal{T}(S)} I(\theta_X=1)p^{(t)}(\theta|S)$
{\em partitioning probability}.
However, it is difficult to compute $\{p_X\}_{X\in I^{(0)}}$ as efficiently as 
in the prediction of secondary structures of RNA sequences, 
where it seems possible to compute the base-pairing probability matrix in polynomial time 
by using dynamic programming).
Instead, a sampling algorithm can be used for estimating $\{p_X\}_{X\in I^{(0)}}$ 
approximately \cite{MetropolisEtAl:53} for this problem.
Once $\{p_X\}_{X\in I^{(0)}}$ is estimated, Theorem~\ref{theorem:g_centroid_th} leads to the following:
\begin{property}[Computaion of Estimator~\ref{est:pt}]
The phylogenetic tree of Estimator~\ref{est:pt} is found by maximizing the 
sum of the partitioning probabilities $p_X$ (of the partitions given by the predicted
tree) that are larger than $1/(\gamma+1)$.
\end{property}

In contrast to Estimator~\ref{est:align} 
(the $\gamma$-centroid estimator for secondary structure prediction of RNA sequence) 
and Estimator~\ref{est:rna_sec_pred} (the $\gamma$-centroid estimator for pairwise alignment),
it appears that there is no efficient method (such as dynamic programming algorithms) to computed 
Estimator~\ref{est:pt} with $\gamma>1$. 
Estimator~\ref{est:align} with $\gamma\in[0,1]$, however,
can be computed by using the following property, which
is directly proven by Corollary~\ref{cor:WC_01} and the definition of the space $\mathcal{T}(S)$.

\begin{property}[Estimator~\ref{est:pt} with $0<\gamma\le 1$]
The $\gamma$-centroid estimator with $\gamma\in[0,1]$ for Problem~\ref{prob:pt}{}
contains its consensus estimator.
\end{property}

\subsubsection{Alignment between two {\em alignments} 
of biological sequences}\label{sec:alignment_align}

In this section we consider the problem of the 
alignment between two multiple alignments of biological sequences 
(Figure~\ref{fig:align_align}), 
which is often important in the multiple alignment of RNA sequences~\cite{pmid15687296}.
This problem is formulated as follows.
\begin{figure}[t]
\centerline{
\includegraphics[width=0.7\textwidth]{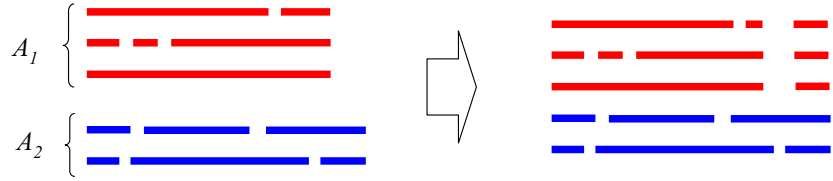} 
}
\caption{\label{fig:align_align}
{\bf Alignment between two multiple alignments $A_1$ and $A_2$ (Problem~\ref{prob:amulti})}
}
\end{figure}
\begin{problem}[Alignment between two alignments of biological sequences]\label{prob:amulti}
The data is represented as $D=\{A,A'\}$ where
${A}$ and $A'$ are alignments of biological sequences
and
the predictive space $Y$ is equal to $\mathcal{A}(A,A')$, that 
is, the space of the alignments of ${A}$ and ${A}'$.
\end{problem}
In the following, $l(A)$ and $n(A)$ denote the length of the alignment 
and the number of sequences in the alignment $A$, respectively.
If both ${A}$ and ${A}'$ contain a single biological sequence (with no gap),
Problem~\ref{prob:amulti} is equivalent to conventional pairwise alignment 
of biological sequences (Problem~\ref{prob:align}).
As in common secondary structure prediction, 
the representative estimator plays an important role in this application.
%

\begin{estimator}[Representative estimator for Problem~\ref{prob:amulti}]\label{est:amulti}
For Problem~\ref{prob:amulti}, 
we obtain the representative estimator (Definition~\ref{def:rep_est}).
The gain function $G'(\theta^k,y)$ is the gain function of the $\gamma$-centroid estimator.
The parameter space $\Theta$ is represented as a product space
$\Theta=\prod_{x\in A, x'\in A'} \mathcal{A}(x,x')$ where $\mathcal{A}(x,x')$
is defined in the previous section.
The probability distribution on the parameter space $\Theta$ is given by 
$p(\theta|D)=\prod_{x\in A, x'\in A'}p^{(a)}(\theta^{xx'}|x,x')$ for
$\theta=(\theta^{xx'})_{x\in A,x'\in A'}\in\Theta$
where $p^{(a)}(\theta|x,x')$ is given in the previous section
(when $x$ or $x'$ contains some gaps, $p^{(a)}(\theta|x,x')$ is defined by
the sequences with the gaps removed).
\end{estimator}
%
Corollary~\ref{cor:gamma_est_for_representative} proves the following properties of Estimator~\ref{est:rep_pred}.
\begin{property}[A Relation of Estimator~\ref{est:amulti} with accuracy measures]
Estimator~\ref{est:amulti} is consistent with the accuracy process for 
Problem~\ref{prob:amulti}
that is shown in Figure~\ref{fig:eval_proc_al_al}. 
We compare every pairwise alignment of $x\in A$ and $x'\in A'$ 
with the reference alignment.
These comparisons are made using TP, TN, FP and FN with respect
to the aligned-bases (e.g., using SEN, PPV and F-score).
\end{property}
\begin{figure}[t]
\centerline{
\includegraphics[width=0.7\textwidth]{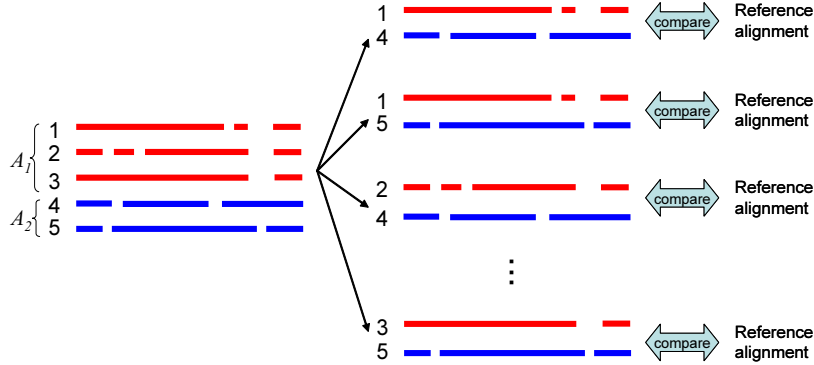} 
}
\caption{\label{fig:eval_proc_al_al}
{\bf An evaluation process for Problem~\ref{prob:amulti}.} 
The comparison between every pairwise
alignment and the reference alignment is conducted using TP, TN, FP and FN with respect
to the aligned-bases.
}
\end{figure}
\begin{property}[Computation of Estimator~\ref{est:amulti}]\label{prop:al_al_th}
Estimator~\ref{est:amulti} can be given by maximizing 
the sum of probabilities $\overline{p_{ik}}$ that are larger than $1/(\gamma+1)$
where
\begin{align}
\overline{p_{ik}} 
=\frac{1}{n(A)n(A')} \sum_{x\in A}\sum_{x'\in A'} \sum_{\theta\in\Theta} 
I(\theta_{ik}=1) p^{(a)}(\theta|x,x'). \label{eq:ap}
\end{align}
Therefore, the pairwise alignment of Estimator~\ref{est:amulti} 
can be computed by the Needleman-Wunsch-type DP algorithm of 
Eq.~(\ref{eq:DP_a_centroid_suppl}) in which we replace $p_{ij}$ with Eq.~(\ref{eq:ap}).
\end{property}
%
%
%
\begin{property}[Computation of Estimator~\ref{est:amulti} with $0\le\gamma\le 1$]
The Estimator~\ref{est:amulti} with $\gamma\in[0,1]$ contains 
the consensus estimator.
Moreover, the consensus estimator is identical to 
the estimator $y=\{y_{ik}^*\}_{1\le i\le l(A),1\le k\le l(A')}$:
\begin{equation*}
y_{ik}^* =\left\{
\begin{array}{ll}
1 & \mbox{if } \overline{p_{ik}}>\frac{1}{\gamma+1}\\
0 & \mbox{if } \overline{p_{ik}}\le \frac{1}{\gamma+1}
\end{array}\right.
\mbox{ for } i=1,2,\ldots,l(A), k=1,2,\ldots,l(A')
\end{equation*}
where $\overline{p_{ik}}$ is defined in Eq.~(\ref{eq:ap}).
\end{property}

The probability matrix $\{\overline{p_{ik}}\}_{1\le i\le l(A),1\le k\le l(A)'}$ is often called
an {\em averaged} aligned-base (matching) probability 
matrix of ${A}$ and ${A}'$.
In the iterative refinement of the \probcons~\cite{pmid15687296} algorithm, 
the existing multiple alignments are randomly partitioned into two groups 
and those two multiple alignments are re-aligned.
This procedure is equivalent to Problem~\ref{prob:amulti}.


The estimator used in \probcons{} 
is identical to Estimator~\ref{est:amulti} in the limit $\gamma\to\infty$.
Therefore, the estimator used in \probcons{} is a special case of Estimator~\ref{est:amulti}
and it only takes into account the SEN or SPS (sum-of-pairs score) 
of a predicted alignment.

\subsubsection{Common secondary structure prediction from a multiple 
alignment of RNA sequences}\label{sec:rna_sec_pred_m}

Common secondary structure prediction from 
a given multiple alignment of RNA sequences plays important role 
in RNA research including
non-coding RNA (ncRNA)~\cite{pmid19833701} and viral RNAs~\cite{pmid19369331},
because it is useful for 
phylogenetic analysis of RNAs~\cite{pmid19723687} and 
gene finding~\cite{pmid19833701, Washietl-2005,pmid16273071,pmid19908361}. 
In contrast to conventional secondary structure prediction of RNA sequences (Problem~\ref{prob:rnas}), 
the input of common
secondary structure prediction is a multiple alignment of RNA sequences and
the output is a secondary
structure whose length is equal to the length of the input alignment (see Figure~\ref{fig:com_sec}).
\begin{figure}[t]
\centerline{
\includegraphics[width=0.7\textwidth]{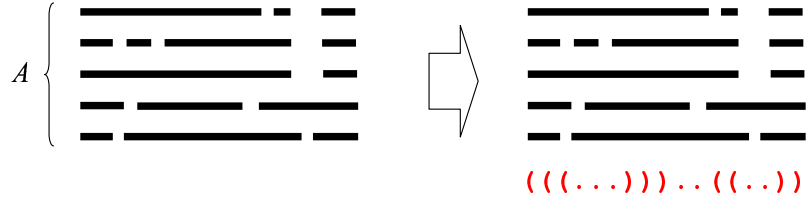} 
}
\caption{\label{fig:com_sec}
{\bf Common secondary structure prediction (Problem~\ref{prob:comm_sec_str_pred})}
}
\end{figure}
\begin{problem}[Common secondary structure prediction]\label{prob:comm_sec_str_pred}
The data is represented as $D=\{A\}$ where $A$ is a multiple alignment of 
RNA sequences and the predictive space $Y$ is identical to $\mathcal{S}(A)$
(the space of secondary structures whose length is equal to the alignment).
\end{problem}

The representative estimator (Definition~\ref{def:rep_est}) directly gives an estimator
for Problem~\ref{prob:comm_sec_str_pred}.
\begin{estimator}[The representative estimator for Problem~\ref{prob:comm_sec_str_pred}]
\label{est:rep_pred}
For Problem~\ref{prob:comm_sec_str_pred}, we obtain the representative estimator
(Definition~\ref{def:rep_est}) as follows.
The gain function $G'(\theta^k,y)$ is the gain function of the $\gamma$-centroid
estimator.  
The parameter space is equal to $\Theta=\prod_{x\in A}\mathcal{S}(x)$ where
$\mathcal{S}(x)$ is the space of secondary structures.
The probability distribution on $\Theta$ is given by
$p(\theta|D)=\prod_{x\in A}p_x(\theta^x|A)$ where $p_x(\theta^x|A)$ is the 
probability distribution of the secondary structures of $x\in A$ after observing
the alignment $A$.
\end{estimator}
For example, $p_x(\theta^x|A)$ can be given by extending the $p^{(s)}(\theta|x)$, although
we have also proposed more appropriate probability distribution
(see~\cite{CentroidAlifold} for the details).

Corollary~\ref{cor:gamma_est_for_representative} proves the following properties of Estimator~\ref{est:rep_pred}.
\begin{property}[A relation of Estimator~\ref{est:rep_pred} with accuracy measures]
Estimator~\ref{est:rep_pred} is consistent with an evaluation process for 
common secondary structure prediction:
First, we map
the predicted common secondary structure into secondary structures in the
multiple alignment, and then 
the mapped structures are compared with the reference secondary structures
based on TP, TN, FP and FN of the base-pairs using, for example, SEN, PPV and MCC
(Figure~\ref{fig:eval_com_sec}).
\end{property}
\begin{figure}[t]
\centerline{
\includegraphics[width=0.7\textwidth]{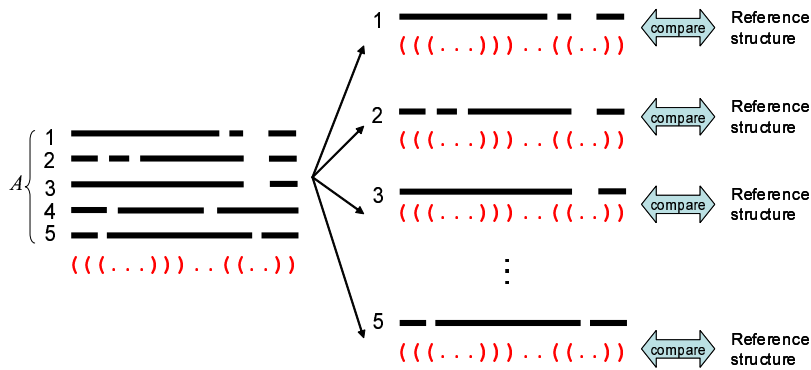} 
}
\caption{\label{fig:eval_com_sec}
{\bf An evaluation process for common secondary structure prediction (Problem~\ref{prob:comm_sec_str_pred}).}
The comparison between each secondary structure
and the reference secondary structure is done using TP, TN, FP and FN with respect
to the base-pairs.
}
\end{figure}
Much research into common secondary structure prediction employs the
evaluation process in Figure~\ref{fig:eval_com_sec} (e.g.,~\cite{pmid19014431}).
%
%
\begin{property}[Computation of Estimator~\ref{est:rep_pred}]
The common secondary structure of Estimator~\ref{est:rep_pred}
is given by maximizing the sum of the averaged 
base-pairing probabilities $\overline{p_{ij}}$ where
\begin{align}
\overline{p_{ij}}=\frac{1}{|A|}\sum_{x\in A} p_x(\theta_{ij}^x=1|A).
\end{align}
Therefore, 
the common secondary structure of the estimator can be computed 
using the dynamic programming algorithm in Eq.~(\ref{eq:dp_centroid_rna_sec}) if we 
replace $p_{ij}$ 
with $\overline{p_{ij}}$.
\end{property}
%
%
Also, we can predict the secondary structure of 
Estimator~\ref{est:rep_pred} without conducting Nussinov-style DP:
\begin{property}[Computation of Estimator~\ref{est:rep_pred} with $0\le\gamma\le 1$]
The secondary structure of Estimator~\ref{est:rep_pred} 
with $\gamma\in[0,1]$ can be predicted by collecting the base-pairs whose averaged base-paring
probabilities are larger than $1/(\gamma+1)$.
\end{property}


It should be noted that the tools of common secondary structure prediction, 
\rnaalifold~\cite{pmid19014431}, \petfold~\cite{pmid18836192} and 
\mccaskillmea~\cite{pmid17182698}
are also considered as a representative estimators 
(Definition~\ref{def:rep_est}). 
In \cite{CentroidAlifold}, the authors systematically discuss those points.
See \cite{CentroidAlifold} for details.

\subsubsection{Pairwise alignment using homologous sequences}\label{sec:alignment_hom}
%
As in the previous application to RNA secondary structure prediction using
homologous sequences, 
if we obtain a set of homologous sequences $H$ for the target sequences 
$x$ and $x'$ 
(see Figure~\ref{fig:align_hom}), 
we would have more accurate estimator for the pairwise alignment of $x$ and $x'$
than Estimator~\ref{est:align}. The problem is formulated
as follows.
\begin{figure}[t]
\centerline{
\includegraphics[width=0.7\textwidth]{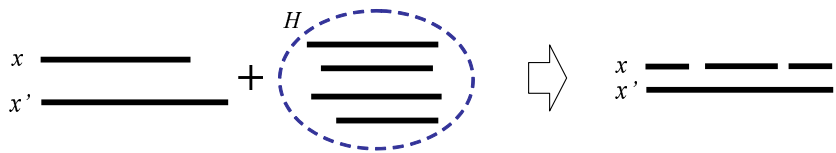} 
}
\caption{\label{fig:align_hom}
{\bf Pairwise alignment using homologous sequences (Problem~\ref{prob:pct})}
}
\end{figure}
%
\begin{problem}[Pairwise alignment using homologous sequences]\label{prob:pct}
%
The data is represented as $D=\{x,x',H\}$ where
$x$ and $x'$ are two biological sequences that we would like to align, 
and $H$ is a set of homologous sequences for $x$ and $x'$.
The predictive space $Y$ is given by $Y = \mathcal{A}(x,x')$ which is 
the space of the pairwise alignments of two sequences $x$ and $x'$.
\end{problem}
The difference between Problem~\ref{prob:align}{} and this problem is that
we can use other biological sequences (that seem to be homologous to $x$ and $x'$) 
besides the two sequences $x$ and $x'$ which are being aligned. 

We can introduce the probability distribution (denoted by $p^{(a)}(\theta|x,x',h)$) on the space of
multiple alignments of three sequences $x$, $x'$ and $h$ (denoted by $\mathcal{A}(x,x',h)$ and
whose definition is similar to that of $\mathcal{A}(x,x')$)
by a model such as the triplet HMM (which is similar to the pair HMM).
Then, we obtain a probability distribution on the space of pairwise alignments of $x$ and $x'$ 
(i.e., $\mathcal{A}(x,x')$) by marginalizing $p^{(a)}(\theta|x,x',h)$ into the space $\mathcal{A}(x,x')$:
\begin{align}
p(\theta|x,x')=\sum_{\theta'\in\Phi^{-1}(\theta)}p^{(a)}(\theta'|x,x',h)\label{eq:marg_ali3}
\end{align}
where $\Phi$ is the projection from $\mathcal{A}(x,x',h)$ into $\mathcal{A}(x,x')$.
Moreover, by averaging these probability distributions over $h\in H$, 
we obtain the following probability distribution on $\mathcal{A}(x,x')$:
\begin{align}
p(\theta|x,x')=\frac{1}{|H|}\sum_{h\in H}\sum_{\theta'\in\Phi^{-1}(\theta)}p^{(a)}(\theta'|x,x',h)\label{eq:p_align_h_ave}
\end{align}
where $|H|$ is the number of sequences in $H$.

The $\gamma$-centroid estimator with the distribution in Eq.~(\ref{eq:p_align_h_ave}) directly 
gives an estimator for Problem~\ref{prob:pct}.
However, to compute the aligned-base-pairs (matching) probabilities $p_{ik}$ 
with respect to this distribution
demands a lot of computational time, so we employ the approximated $\gamma$-type estimator (Definition~\ref{est:approx}) of
this $\gamma$-centroid estimator as follows.

\begin{estimator}[Approximated $\gamma$-type estimator for Problem~\ref{prob:pct}]\label{est:PE_pct}
We obtain the approximated $\gamma$-type estimator 
(Definition~\ref{est:approx}) for Problem~\ref{prob:pct} with the following settings.
The parameter space is given by $\Theta=\Theta'\times\Theta'^\perp$ where
\begin{align*}
\Theta'=\mathcal{A}(x,x') (=Y) \mbox{ and } \Theta'^\perp=\prod_{h\in H}[\mathcal{A}(x,h)\times\mathcal{A}(x',h)]
\end{align*}
and the probability distribution on the parameter space $\Theta'$ is defined by
\begin{align}
p(\theta\left|D\right.)
= p^{(a)}(\theta^{xx'}|x,x')\prod_{h\in H} \left[p^{(a)}(\theta^{xh}|x,h)
p^{(a)}(\theta^{x'h}|x',h)\right]
\end{align}
for 
$\theta=(\theta^{xx'}, \{\theta^{xh}, \theta^{x'h}\}_{h\in H})\in\Theta=\Theta'\times\Theta'^\perp$.
The pointwise gain function (see Definition~\ref{def:pgain_function}) in Eq.~(\ref{def:app_gamma_gain}) is defined by 
\begin{align}
\delta_{ik}(\theta)= \frac{1}{1+|H|}\biggl\{I(\theta^{xx'}_{ik}=1) + \sum_{h\in H}\sum_{v=1}^{|h|} I(\theta_{iv}^{xh}=1)I(\theta_{kv}^{x'h}=1)\biggr\}
\end{align}
where $|h|$ is the length of the sequence $h$.
\end{estimator}
\begin{property}[Computation of Estimator~\ref{est:PE_pct}]
The alignment of Estimator~\ref{est:PE_pct} is equal to the alignment that maximizes
the sum of $p_{ik}$ larger than $1/(\gamma+1)$ where
\begin{align} 
p_{ik}= \frac{1}{|H|+1}\biggl\{p(\theta^{xx'}_{ik}=1|x,x')+\sum_{h\in H}\sum_{v=1}^{|h|} 
p^{(a)}\left(\theta^{xh}_{iv}=1|x,h\right)p^{(a)}\left(\theta^{x'h}_{kv}=1\Big|x',h\right)\biggr\}.\label{eq:PCT_prob}
\end{align}
Therefore, the recursive equation of the dynamic program to calculate
the alignment of Estimator~\ref{est:PE_pct} is given by replacing 
$p_{ik}$  in Eq.~(\ref{eq:DP_a_centroid_suppl}) with 
Eq.~(\ref{eq:PCT_prob}).
\end{property}

Moreover, by using Theorem~\ref{theorem:general}, we have the following proposition,
which enables us to compute the proposed estimator for $\gamma\in[0,1]$ without
using (Needleman-Wunsch-type) dynamic programming.
%
\begin{property}[Computation of Estimator~\ref{est:PE_pct} for $0\le\gamma\le 1$]\label{prop:PCT}
The pairwise alignment of Estimator~\ref{est:PE_pct} with $\gamma \in [0,1]$ can be
predicted by collecting the aligned-bases
whose probability $p_{ik}$ in (\ref{eq:PCT_prob}) is larger than $1/(\gamma+1)$.
\end{property}

It should be noted that 
$\{p_{ik}\}_{1\le i\le|x|,1\le k\le|x'|}$ is identical to 
the {\em probability consistency transformation} (PCT) 
of $x$ and $x'$~\cite{pmid15687296}.
In \probcons~\cite{pmid15687296}, 
the pairwise alignment is predicted by the Estimator~\ref{est:PE_pct} with sufficiently 
large $\gamma$.
Therefore, the estimator for Problem~\ref{prob:pct} 
used in the \probcons{} algorithm is a special case of Estimator~\ref{est:PE_pct}.

\subsubsection{RNA secondary structure prediction using homologous sequences}\label{sec:rna_sec_pred_hom}
%
If we obtain a set of homologous RNA sequences for the target RNA sequence, we might have a more 
accurate estimator~\cite{pmid19478007} for secondary structure prediction 
than the $\gamma$-centroid estimator (Estimator~\ref{est:rna_sec_pred}).
This problem is formulated as follows and was considered in~\cite{pmid19478007} for the first time
{(See Figure~\ref{fig:sec_hom})}.
\begin{figure}[t]
\centerline{
\includegraphics[width=0.7\textwidth]{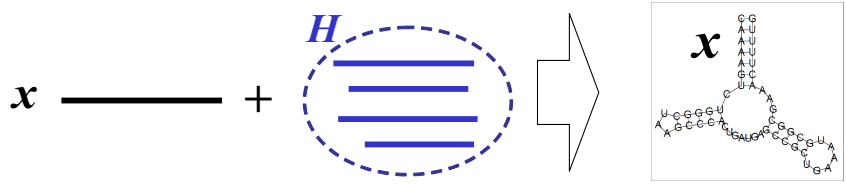} 
}
\caption{\label{fig:sec_hom}
{\bf RNA secondary structure prediction using homologous sequences (Problem~\ref{prob:rna_sec_pred_hom})}
}
\end{figure}
\begin{problem}[RNA secondary structure prediction using homologous sequences]\label{prob:rna_sec_pred_hom}
The data $D$ is represented as $D=\{x,H\}$ where
$x$ is the target RNA sequence for which we would like to make secondary structure predictions and 
$H$ is the set of its homologous sequences.
The predictive space $Y$ is identical to $\mathcal{S}(x)$, the space of 
the secondary structures of an RNA sequence $x$.
\end{problem}
The difference between this problem and Problem~\ref{prob:rnas}{} is that we are able to
employ homologous sequence information 
for predicting the secondary structure of the target RNA sequence.
In this problem, it is natural that we assume 
the target sequence $x$ and each homologous sequence $h\in H$
share {\em common} secondary structures.
The common secondary structure is naturally modeled by a {\em structural} 
alignment (that considers not only the alignment between bases but also the alignment between
base-pairs), 
and the probability distribution (denoted by $p^{(sa)}(\theta|x,x')$)
on the space of the structural alignments of two RNA sequences $x$ and $x'$ (denoted by $\mathcal{SA}(x,x')$)
is given by the Sankoff model~\cite{sankoff1985}.
By marginalizing the distribution $p^{(sa)}$ into the space of secondary structures $\mathcal{S}(x)$ 
of the target sequence $x$, we obtain more reliable distribution $p(\theta|x)$ on $\mathcal{S}(x)$:
\begin{align}
p(\theta|x)=\sum_{\theta'\in\Phi^{-1}(\theta)}p^{(sa)}(\theta'|x,h)\label{eq:marg_p}
\end{align}
where $\Phi$ is the projection from $\mathcal{SA}(x,h)$ into $\mathcal{S}(x)$.
Moreover, by averaging these probability distributions on $\mathcal{S}(x)$, 
we obtain the following probability 
distribution of secondary structures of the target sequence.
\begin{align}
p(\theta|x)=\frac{1}{|H|}\sum_{h\in H}\sum_{\theta'\in\Phi^{-1}(\theta)}p^{(sa)}(\theta'|x,h)\label{eq:ave_marg_p}
\end{align}
where $|H|$ is the number of sequences in $H$.
The $\gamma$-centroid estimator with the probability distribution in Eq.~(\ref{eq:ave_marg_p}) gives a
reasonable estimator for Problem~\ref{prob:rna_sec_pred_hom}, because Eq.~(\ref{eq:ave_marg_p}) considers
consensus secondary structures between $x$ and $h\in H$.
However, the calculation of the $\gamma$-estimator requires huge computational cost
because it requires $O(nL^6)$ for computing the base-paring probability matrix $\{p_{ik}\}$
where $p_{ik}=\sum_{\theta\in\mathcal{S}(x)} I(\theta_{ij}=1)p(\theta|x)$ with 
the distribution of Eq.~(\ref{eq:ave_marg_p}). 
Therefore, we employ
the approximated $\gamma$-type estimator (Definition~\ref{est:approx}) of the $\gamma$-centroid estimator,
which is equivalent to the estimator proposed in~\cite{pmid19478007}.
\begin{estimator}[Approximated $\gamma$-type estimator for Problem~\ref{prob:rna_sec_pred_hom}]
\label{est:approx_est_rnas_hom}
We obtain the approximated $\gamma$-type estimator (Definition~\ref{est:approx}) 
for Problem~\ref{prob:rna_sec_pred_hom} with
the following settings. 
The parameter space is given by $\Theta=\Theta'\times\Theta'^\perp$ where
\begin{align*}
&\Theta'=\mathcal{S}(x)(=Y) \mbox{ and }
\Theta'^\perp=\prod_{h\in H}[\mathcal{A}(x,h)\times\mathcal{S}(h)],
\end{align*}
and the probability distribution on $\Theta$ is defined by
\begin{align*}
&p(\theta|D)=p^{(s)}(\theta^x|x)\prod_{h\in H}\left[p^{(a)}(\theta^{xh}|x,h)p^{(s)}(\theta^h|h)\right]
\end{align*}
for
$\theta=(\theta^x,\{\theta^{xh},\theta^h\}_{h\in H})\in\Theta=\Theta'\times\Theta'^\perp$.
Moreover, Eq.~(\ref{def:app_gamma_gain}) in the pointwise gain function is defined by 
\begin{align*}
&\delta_{ij}(\theta)=
\alpha I(\theta^{x}_{ij}=1) - \frac{1-\alpha}{|H|}\sum_{h\in H}
\sum_{k<l} I(\theta^{xh}_{ik}=1) I(\theta^{xh}_{jl}=1) I(\theta^{h}_{kl}=1)
\end{align*}
for $\alpha\in[0,1]$.
\end{estimator}
It should be noted that Estimator~\ref{prob:rna_sec_pred_hom} 
is equivalent to the estimator proposed in~\cite{pmid19478007}.
The secondary structure of the estimator can be computed by the following method.
\begin{property}[Computation of Estimator~\ref{est:approx_est_rnas_hom}]
The secondary structure of Estimator~\ref{est:approx_est_rnas_hom} is computed by
maximizing the sum of $p_{ij}$ larger than $1/(\gamma+1)$ where
\begin{align}
&p_{ij}= \alpha p^{(s,x)}_{ij} + \frac{1-\alpha}{|H|}\sum_{h\in H}\sum_{k<l} 
p^{(a,x,h)}_{ik,jl}
p^{(s,h)}_{kl}\label{eq:marg_prob_rnas_hom}.
\end{align}
Here, $p^{(s,x)}_{ij}=p^{(s)}(\theta^x_{ij}=1|x)$ and
$p^{(a,x,h)}_{ik,jl}=p^{(a)}(\theta^{xh}_{ik}=1, \theta^{xh}_{jl}=1|x,h)$.
Therefore, the secondary structure of Estimator~\ref{est:approx_est_rnas_hom}
can be computed by the Nussinov-type DP of 
Eq.~(\ref{eq:dp_centroid_rna_sec}) in which we replace $p_{ij}$ by Eq.~(\ref{eq:marg_prob_rnas_hom}).
\end{property}

The computational cost with respect to time for computing the secondary structure
of Estimator~\ref{est:approx_est_rnas_hom} 
is $O(n L^4)$ where $n$ is the number of RNA sequences and $L$ is the length 
of RNA sequences.
In~\cite{pmid19478007}, we employed a further approximation of the estimator, 
and reduced the computational cost to $O(n L^3)$.
We implemented this estimator in software called \centroidhomfold.
See~\cite{pmid19478007} for details of the theory and results of computational 
experiments.
Although the authors did not mention it in their paper~\cite{pmid19478007}, 
the following property holds.
\begin{property}[Computation of Estimator~\ref{est:approx_est_rnas_hom} with $0\le\gamma\le 1$]
Estimator~\ref{est:approx_est_rnas_hom} with $\gamma\in[0,1]$ 
can be predicted by collecting the aligned-bases where the 
(pseudo-)base-paring probability of Eq.~(\ref{eq:marg_prob_rnas_hom}) 
is larger than $1/(\gamma+1)$.
\end{property}

%
\subsubsection{Pairwise alignment of {\em structured} RNAs}\label{sec:rna_alignment}
%
In this section, we focus on the pairwise alignment of structured RNAs. 
This problem is formulated as Problem~\ref{prob:align}, so the output of the problem is 
a usual alignment (contained in $\mathcal{A}(x,x')$).
In contrast to the usual alignment problem, we can consider
not only nucleotide sequences but also secondary structures in each sequence 
for the problem. 
Note that this does {\em not} mean the structural alignment~\cite{sankoff1985} 
of RNA sequences, because
the structural alignment produces both alignment and the common secondary structure
simultaneously.

The probability distributions $p^{(a)}(\theta|x,x')$ on $\mathcal{A}(x,x')$ 
described in the previous section
are not able to handle secondary structures of each RNA sequence.
In order to obtain a probability distribution on $\mathcal{A}(x,x')$ 
that considers secondary structure, we employ the marginalization of the
Sankoff model~\cite{sankoff1985}
that gives a probability distribution (denoted by $p^{(sa)}(\theta|x,x')$) 
on the space of possible structural alignments between two RNA sequences (denoted by 
$\mathcal{SA}(x,x')$).
In other words, 
we obtain a probability distribution on the space $\mathcal{A}(x,x')$
by marginalizing the probability distribution of {\em structural} alignments of 
two RNA sequences (given by the Sankoff model) 
into the space $\mathcal{A}(x,x')$ as follows.
\begin{align}
p(\theta|x,x')=\sum_{\theta'\in\Phi^{-1}(\theta)}p^{(sa)}(\theta'|x,x')\label{eq:m_p_align}
\end{align}
where $\Phi$ is the projection from $\mathcal{SA}(x,x')$ into $\mathcal{A}(x,x')$,  
$\theta\in\mathcal{A}(x,x')$ and $\theta'\in\mathcal{SA}(x,x')$.
The difference between this marginalized probability distribution and the distributions
such as Miyazawa model
is that the former considers secondary structures of 
each sequence (more precisely, the former considers the common secondary structure).

Then, the $\gamma$-centroid estimator with this distribution Eq.~(\ref{eq:m_p_align}) will 
give a 
reasonable estimator for the pairwise alignment of two RNA sequences.
However, the computation of this estimator demands huge computational cost because it uses the
Sankoff model (cf. it requires $O(L^6)$ time for computing the matching probability 
matrix of structural alignments). 
Therefore, we employed the approximated $\gamma$-type estimator (Definition~\ref{est:approx}) of 
the $\gamma$-centroid estimator with the marginalized distribution as follows.
\begin{estimator}[Approximated $\gamma$-type estimator for Problem~\ref{prob:align}{} with two RNA sequences]\label{est:al_rna}
In Problem~\ref{prob:align} where $x$ and $x'$ are RNA sequences,
we obtain the approximated $\gamma$-type estimator (Estimator 2) with the following settings.
The parameter space is given by $\Theta=\Theta'\times\Theta'^\perp$ where 
\begin{align*}
&\Theta' = \mathcal{A}(x,x') (=Y), \quad \Theta'^\perp=\mathcal{S}(x) \times \mathcal{S}(x')
\end{align*}
and the probability distribution on the parameter space $\Theta$ is defined by
\begin{align*}
&p(\theta|x,x') = 
p^{(a)}(\theta^{(a,x,x')}|x,x')p^{(s)}(\theta^{(s,x)}|x)p^{(s)}(\theta^{(s,x')}|x')
\end{align*}
for $\theta=(\theta^{(a,x,x')}, \theta^{(s,x)}, \theta^{(s,x')}) \in \Theta$.
The pointwise gain function of Eq.~(\ref{def:app_gamma_gain}) is defined by 
\begin{align*}
\delta_{uv}(\theta) = w_1 \theta_{uv}^{(a,x,x')} + 
w_2
\left(
\overline{R}_{uv}(\theta)+\overline{L}_{uv}(\theta') 
\right)+
w_3
\eta_{u}^{(x)} \eta_{v}^{(x')}
\end{align*}
where
\begin{align*}
\overline{R}_{uv}(\theta):=
\sum_{j:u<j,l:v<l} \theta_{uj}^{(s,x)} \theta_{vl}^{(s,x')} \theta_{jl}^{(a,x,x')}, \\
\overline{L}_{uv}(\theta):=
\sum_{i:i<u, k:k<v} \theta_{iu}^{(s,x)} \theta_{kv}^{(s,x')} \theta_{ik}^{(a,x,x')}, \\ 
\eta_{u}^{(x)}:=\prod_{j:u<j} (1-\theta^{(s,x)}_{uj})\prod_{j:j<u} (1-\theta^{(s,x)}_{ju}),
\end{align*}
and $w_1$, $w_2$ and $w_3$ are positive weights that satisfy
$w_1+w_2+w_3=1$.
\end{estimator}

This approximated $\gamma$-type estimator is
equivalent to the estimator proposed in~\cite{pmid19808876} and the alignment of the 
estimator can be computed by the following property.
\begin{property}[Computation of Estimator~\ref{est:al_rna}]
The alignment of Estimator~\ref{est:al_rna} can be computed by maximizing the sum of probabilities 
$p_{uv}$ that are larger than $1/(\gamma+1)$ where
\begin{align}
&p_{uv} = w_1 p_{uv}^{(a,x,x')} +  \nonumber\\
&\qquad w_2 \Bigl( \sum_{j:u<j, l:v<l} p_{uj}^{(s,x)} p_{vl}^{(s,x')} p_{jl}^{(a,x,x')} +
\sum_{i:i<u, k:k<v} p_{iu}^{(s,x)} p_{kv}^{(s,x')} p_{ik}^{(a,x,x')} \Bigr) + 
w_3 q_u^{(s,x)} q_v^{(s,x')}.\label{eq:pseudo_prob_align_st_rna}
\end{align}
Here, we define
\begin{align*}
&p_{ij}^{(s,x)}=\sum_{\theta\in\mathcal{S}(x)} \theta_{ij} p^{(s)}(\theta|x),\\
&q_u^{(s,x)}=1-\sum_{i:i<u} p_{iu}^{(s,x)}-\sum_{j:u<j}p_{uj}^{(s,x)} \mbox{ and }\\
&p_{uv}^{(a,x,x')}=\sum_{\theta\in\mathcal{A}(x,x')} \theta_{uv} p^{(a)}(\theta|x,x').
\end{align*}
Therefore, the pairwise alignment of Estimator~\ref{est:al_rna} 
can be computed by a Needleman-Wunsch-type dynamic program of 
Eq.~(\ref{eq:DP_a_centroid_suppl}) in which we replace $p_{ij}$ 
with Eq.~(\ref{eq:pseudo_prob_align_st_rna}).
\end{property}
Note that $p_{uv}$ in Eq.~(\ref{eq:pseudo_prob_align_st_rna}) is considered as a
{\em pseudo}-aligned base probability where $x_u$ aligns with $x_v$.

By checking Eq.~(\ref{eq:condition_ae}), we obtain the following property:
\begin{property}[Computation of Estimator~\ref{est:al_rna} with $0\le\gamma\le 1$]
The pairwise alignment of Estimator~\ref{est:al_rna} can be predicted by collecting
aligned-bases where the
probability in  Eq.~(\ref{eq:pseudo_prob_align_st_rna}) is larger than $1/(\gamma+1)$.
\end{property}

\subsection{Proofs}\label{sec:proof}

In this section, we give the proofs of the theorems, propositions 
and corollary.

\subsubsection{Proof of Theorem~\ref{theorem:general}}

We will prove a more general case of Theorem~\ref{theorem:general} where the parameter space $\Theta$
is different from the predictive space $Y$ and a probability distribution on $\Theta$
is assumed (cf. Assumption~\ref{as:our}).

\begin{theorem}\label{theorem:general_suppl}
In Problem~\ref{prob:1} with Assumption~\ref{as:cl} 
and a pointwise gain function,
suppose that a predictive space $Y$ can be written as
\begin{equation}
Y = \bigcap_{k=1}^K C_k,\label{eq:constrain_Y_suppl}
\end{equation}
where $C_k$ is defined as
\begin{equation*}
C_k =\biggl\{ y \in \{0,1\}^n\bigg|\sum_{i \in I_k} y_i\le 1\biggr\}
\mbox{ for } k=1,2,\ldots,K\label{eq:constrain_suppl}
\end{equation*}
for an index-set $I_k \subset \{1,2,\ldots,n\}$.
If the pointwise gain function in Eq.~(\ref{eq:gain_function}) (we here think $\theta$ is in a parameter space 
$\Theta$ which might be different from $Y$) satisfies the condition
\begin{equation}
F_i(\theta,1)-F_i(\theta,0)+F_j(\theta,1)-F_j(\theta,0) \le 0\label{eq:zyouken_suppl}
\end{equation}
for every $\theta\in \Theta$ and every $i,j \in I_k$ ($1\le k \le K$),
then the consensus estimator is in the predictive space $Y$,
and hence the MEG estimator contains the consensus estimator.
\end{theorem}

\begin{proof}
It is sufficient to show that the consensus estimator $\hat y^{(c)}$
is contained in the predictive space $Y$ because 
$\bar{G}(\hat y) \le \bar{G}(\hat y^{(c)})$ for all $\hat y$ in the MEG estimators, 
where
\begin{eqnarray*}
\overline{G}(y) := E_{\theta|D} [G(\theta,y)]=\int G(\theta, y) p(\theta|D)d\theta.
\end{eqnarray*}
If we assume that $\hat y^{(c)}$
is {\em not} contained in the predictive space, $Y$ that is, 
$\hat y^{(c)} \not \in Y$, then there exists a $k_0$ such that
$\hat y^{(c)} \not \in C_{k_0}$. 
Because $\hat y^{(c)}$ is a binary vector, 
there exist indexes $i,j\in I_{k_0}$ such that $i \ne j$, 
$\hat y^{(c)}_i=1$ and $\hat y^{(c)}_j=1$.
By the definition of $\hat y^{(c)}$, we obtain
\begin{equation*}
E\left[F_i(\theta,1)\right] > E\left[F_i(\theta,0)\right] \mbox{ and }
E\left[F_j(\theta,1)\right] > E\left[F_j(\theta,0)\right].
\end{equation*}
Therefore, we obtain
\begin{align*}
0&<E\left[F_i(\theta,1)-F_i(\theta,0)+F_j(\theta,1)-F_j(\theta,0)\right] \\
&= \int \left[F_i(\theta,1)-F_i(\theta,0)+F_j(\theta,1)-F_j(\theta,0)\right]p(\theta|D)d\theta\\
&\le 0.
\end{align*}
In order to prove the last inequality, we use Eq.~(\ref{eq:constrain}).
This leads to a contradiction and the theorem is proved.
\end{proof}
\begin{remark}
It should be noted that the above theorem holds for an arbitrary parameter space
including continuous-valued spaces.
\end{remark}

\subsubsection{Proof of Theorem~\ref{thm:impl_g_est}}\label{sec:proof_impl_g_est}
%
\begin{proof}
Because $I(y_{i}=1)+I(y_{i}=0)=1$ for arbitrary $i$, 
we obtain, using the definitions given in equations 
(\ref{eq:tp}),(\ref{eq:tn}),(\ref{eq:fp}) and (\ref{eq:fn}),
\begin{align*}
&TP+FN=\sum_{i}I(\theta_{i}=1) \mbox{ and } TN+FP=\sum_{i}I(\theta_{i}=0).
\end{align*}
Therefore, we have
\begin{align*}
\alpha_1 TP &+ \alpha_2 TN - \alpha_3 FP - \alpha_4 FN \\
&= (\alpha_1+\alpha_4) TP + (\alpha_2+\alpha_3) TN 
-\alpha_3\sum_{i}I(\theta_{i}=0) -\alpha_4 \sum_{i}I(\theta_{i}=1)\\
&= (\alpha_2+\alpha_3) \left( \frac{\alpha_1+\alpha_4}{\alpha_2+\alpha_3} TP +  TN \right) 
-\alpha_3\sum_{i}I(\theta_{i}=0) -\alpha_4 \sum_{i}I(\theta_{i}=1)
\end{align*}
and this leads to the proof of the theorem. 
\end{proof}

\subsubsection{Proof of Theorem~\ref{theorem:g_centroid_th}}

\begin{proof}
The expectation of the gain function of the $\gamma$-centroid estimator 
is computed as
\begin{align*}
E_{\theta|D}[G(\theta,y)] 
&= \sum_{\theta\in \Theta} \sum_{i=1}^n \left[\gamma I(\theta_i=1)I(y_i=1)+I(\theta_i=0)I(y_i=0)\right] p(\theta|D)\\
&= \sum_{i=1}^n \left[\gamma\cdot p_i \cdot I(y_i=1)+(1-p_i)(1-I(y_i=1))\right]\\
&= \sum_{i=1}^n\left[(\gamma+1)p_i-1\right] I(y_i=1) + \sum_i (1-p_i)
\end{align*}
where $p_i=p(\theta_i=1|D)=\sum_{\theta\in\Theta} I(\theta_i=1)p(\theta|D)$ is 
the marginalized probability.
Therefore, we should always predict $y_i=0$ whenever $p_i < 1 / (\gamma+1)$, because
the assumption of Theorem~\ref{theorem:g_centroid_th} ensures that 
the prediction $y_i=0$ never violate the condition of the predictive space $Y$.
Theorem~\ref{theorem:g_centroid_th} follows by using those facts. 
\end{proof}


\subsubsection{Proof of Corollary~\ref{cor:WC_01}}

%
\begin{proof}
For every $\theta\in\Theta$, $k = 1,2,\ldots, K$, $i,j\in J_{k}$, $\gamma\in[0,1]$, 
we have
\begin{align*}
F_i(\theta,1) & -F_i(\theta,0)+F_j(\theta,1)-F_j(\theta,0) \nonumber\\
&= \gamma I(\theta_i=1)-I(\theta_i=0)+\gamma I(\theta_j=1)-I(\theta_j=0)\\
&\le 2\left(I(\theta_i=1)+I(\theta_j=1)\right)-2\\
&\le 0
\end{align*}
and the condition of Eq.~(\ref{eq:zyouken}) in Theorem~\ref{theorem:general} is satisfied
(in order to prove the last inequality, we use $I(\theta_i=1)+I(\theta_j=1)\le 1$ 
because $i,j \in J_k$). 
Therefore, by Theorem~\ref{theorem:general},
the $\gamma$-centroid estimator contains its consensus estimator.

The last half of the corollary is easily proved using the equation
\begin{equation*}
\sum_{\theta\in\Theta} F_i(\theta,y_i) p(\theta|D)
= \sum_{\theta\in\Theta} \left(I(\theta_i=y_i=0)+\gamma I(\theta_i=y_i=1)\right)p(\theta|D)
=\left\{\begin{array}{ll}
\gamma p_i & \mbox{ for } y_i =1\\
1-p_i & \mbox{ for } y_i =0
\end{array}\right.
\end{equation*}
where $p_i=p(\theta_i=1|D)=\sum_{\theta\in\Theta}I(\theta_i=1)p(\theta|D)$.
\end{proof}

\subsubsection{Proof of Proposition~\ref{prop:equiv_common_estimator}}

\begin{proof}
The representative estimator in Definition~\ref{def:rep_est} can be written as
\begin{align*}
\hat y 
&= \argmax_{y\in Y} \int G(\theta,y)p(\theta|D)d\theta\\
&= \argmax_{y\in Y} \int \biggl[ \sum_{k=1}^K G'(\theta^k,y) \biggr] \biggl[ \prod_{k=1}^K p^{(k)}(\theta^k|D)\biggr] d\theta\\
&= \argmax_{y\in Y} \int G'(\theta',y) \biggl[ \frac{1}{K}\sum_{k=1}^K p^{(k)}(\theta'|D)\biggr] d\theta'
\end{align*}
Then, we finish the proof of Proposition~\ref{prop:equiv_common_estimator}.
\end{proof}





\subsubsection{Derivation of Eq.~(\ref{eq:condition_ae})}

The equation is easily derived from the equality
$F_i(\theta',1)-F_i(\theta',0)=(\gamma+1)\delta_i(\theta')-1$.


%
%



\bibliographystyle{plos2009}
\bibliography{bib-paper}

\begin{thebibliography}{10}
\providecommand{\url}[1]{\texttt{#1}}
\providecommand{\urlprefix}{URL }
\expandafter\ifx\csname urlstyle\endcsname\relax
  \providecommand{\doi}[1]{doi:\discretionary{}{}{}#1}\else
  \providecommand{\doi}{doi:\discretionary{}{}{}\begingroup
  \urlstyle{rm}\Url}\fi
\providecommand{\bibAnnoteFile}[1]{%
  \IfFileExists{#1}{\begin{quotation}\noindent\textsc{Key:} #1\\
  \textsc{Annotation:}\ \input{#1}\end{quotation}}{}}
\providecommand{\bibAnnote}[2]{%
  \begin{quotation}\noindent\textsc{Key:} #1\\
  \textsc{Annotation:}\ #2\end{quotation}}
\providecommand{\eprint}[2][]{\url{#2}}

\bibitem{centroid}
Carvalho L, Lawrence C (2008) {{C}entroid estimation in discrete
  high-dimensional spaces with applications in biology}.
\newblock Proc Natl Acad Sci USA 105: 3209--3214.
\bibAnnoteFile{centroid}

\bibitem{pmid19478997}
Bradley RK, Roberts A, Smoot M, Juvekar S, Do J, et~al. (2009) {{F}ast
  statistical alignment}.
\newblock PLoS Comput Biol 5: e1000392.
\bibAnnoteFile{pmid19478997}

\bibitem{pmid18796475}
Bradley RK, Pachter L, Holmes I (2008) {{S}pecific alignment of structured
  {R}{N}{A}: stochastic grammars and sequence annealing}.
\newblock Bioinformatics 24: 2677--2683.
\bibAnnoteFile{pmid18796475}

\bibitem{Frith_BMCB}
Frith MC, Hamada M, Horton P (2010) {{P}arameters for accurate genome
  alignment}.
\newblock BMC Bioinformatics 11: 80.
\bibAnnoteFile{Frith_BMCB}

\bibitem{pmid16873527}
Do C, Woods D, Batzoglou S (2006) {{C}{O}{N}{T}{R}{A}fold: {R}{N}{A} secondary
  structure prediction without physics-based models}.
\newblock Bioinformatics 22: e90--98.
\bibAnnoteFile{pmid16873527}

\bibitem{pmid19703939}
Lu ZJ, Gloor JW, Mathews DH (2009) {{I}mproved {R}{N}{A} secondary structure
  prediction by maximizing expected pair accuracy}.
\newblock RNA 15: 1805--1813.
\bibAnnoteFile{pmid19703939}

\bibitem{pmid17182698}
Kiryu H, Kin T, Asai K (2007) {{R}obust prediction of consensus secondary
  structures using averaged base pairing probability matrices}.
\newblock Bioinformatics 23: 434--441.
\bibAnnoteFile{pmid17182698}

\bibitem{pmid18836192}
Seemann S, Gorodkin J, Backofen R (2008) {{U}nifying evolutionary and
  thermodynamic information for {R}{N}{A} folding of multiple alignments}.
\newblock Nucleic Acids Res 36: 6355--6362.
\bibAnnoteFile{pmid18836192}

\bibitem{pmid15961464}
Kall L, Krogh A, Sonnhammer EL (2005) {{A}n {H}{M}{M} posterior decoder for
  sequence feature prediction that includes homology information}.
\newblock Bioinformatics 21 Suppl 1: i251--257.
\bibAnnoteFile{pmid15961464}

\bibitem{RNAint}
Kato Y, Sato K, Hamada M, Watanabe Y, Asai K, et~al. (2010) {{R}act{I}{P}: fast
  and accurate prediction of {R}{N}{A}-{R}{N}{A} interaction using integer
  programming}.
\newblock Bioinformatics 26: i460--466.
\bibAnnoteFile{RNAint}

\bibitem{pmid18096039}
Gross S, Do C, Sirota M, Batzoglou S (2007) {{C}{O}{N}{T}{R}{A}{S}{T}: a
  discriminative, phylogeny-free approach to multiple informant de novo gene
  prediction}.
\newblock Genome Biol 8: R269.
\bibAnnoteFile{pmid18096039}

\bibitem{centroidfold-submit}
Hamada M, Kiryu H, Sato K, Mituyama T, Asai K (2009) {{P}rediction of {R}{N}{A}
  secondary structure using generalized centroid estimators}.
\newblock Bioinformatics 25: 465--473.
\bibAnnoteFile{centroidfold-submit}

\bibitem{pmid8771180}
Miyazawa S (1995) {{A} reliable sequence alignment method based on
  probabilities of residue correspondences}.
\newblock Protein Eng 8: 999--1009.
\bibAnnoteFile{pmid8771180}

\bibitem{pmid1695107}
McCaskill JS (1990) {The equilibrium partition function and base pair binding
  probabilities for RNA secondary structure}.
\newblock Biopolymers 29: 1105--1119.
\bibAnnoteFile{pmid1695107}

\bibitem{pmid19435882}
Sato K, Hamada M, Asai K, Mituyama T (2009)
  {{C}{E}{N}{T}{R}{O}{I}{D}{F}{O}{L}{D}: a web server for {R}{N}{A} secondary
  structure prediction}.
\newblock Nucleic Acids Res 37: W277--280.
\bibAnnoteFile{pmid19435882}

\bibitem{MetropolisEtAl:53}
Metropolis N, Rosenbluth A, Teller M, Teller E (1953) Equations of state
  calculations by fast computing machine.
\newblock J Chem Phys 21: 1087-1091.
\bibAnnoteFile{MetropolisEtAl:53}

\bibitem{citeulike:515472}
Robinson DF, Foulds LR (1981) Comparison of phylogenetic trees.
\newblock Mathematical Biosciences 53: 131--147.
\bibAnnoteFile{citeulike:515472}

\bibitem{CWT}
Iwasaki W, Takagi T (2010) {{A}n intuitive, informative, and most balanced
  representation of phylogenetic topologies}.
\newblock Syst Biol 59: 584--593.
\bibAnnoteFile{CWT}

\bibitem{pmid15687296}
Do C, Mahabhashyam M, Brudno M, Batzoglou S (2005) {{P}rob{C}ons:
  {P}robabilistic consistency-based multiple sequence alignment}.
\newblock Genome Res 15: 330--340.
\bibAnnoteFile{pmid15687296}

\bibitem{pmid10871264}
Baldi P, Brunak S, Chauvin Y, Andersen CA, Nielsen H (2000) {{A}ssessing the
  accuracy of prediction algorithms for classification: an overview}.
\newblock Bioinformatics 16: 412--424.
\bibAnnoteFile{pmid10871264}

\bibitem{schwartz-2005}
Schwartz AS, Myers EW, Pachter L (2005).
\newblock Alignment metric accuracy.
\bibAnnoteFile{schwartz-2005}

\bibitem{pmcc}
Hamada M, Sato K, Asai K (2010) {Prediction of RNA secondary structure by
  maximizing pseudo-expected accuracy}.
\newblock BMC Bioinformatics 11: 586.
\bibAnnoteFile{pmcc}

\bibitem{pmid19478007}
Hamada M, Sato K, Kiryu H, Mituyama T, Asai K (2009) {{P}redictions of
  {R}{N}{A} secondary structure by combining homologous sequence information}.
\newblock Bioinformatics 25: i330--338.
\bibAnnoteFile{pmid19478007}

\bibitem{pmid16043502}
Ding Y, Chan C, Lawrence C (2005) {{R}{N}{A} secondary structure prediction by
  centroids in a {B}oltzmann weighted ensemble}.
\newblock RNA 11: 1157--1166.
\bibAnnoteFile{pmid16043502}

\bibitem{CentroidAlifold}
Hamada M, Sato K, Asai K (2010) {{I}mproving the accuracy of predicting
  secondary structure for aligned {R}{N}{A} sequences}.
\newblock Nucleic Acids Res : doi: 10.1093/nar/gkq792.
\bibAnnoteFile{CentroidAlifold}

\bibitem{pmid19808876}
Hamada M, Sato K, Kiryu H, Mituyama T, Asai K (2009) {{C}entroid{A}lign: fast
  and accurate aligner for structured {R}{N}{A}s by maximizing expected
  sum-of-pairs score}.
\newblock Bioinformatics 25: 3236--3243.
\bibAnnoteFile{pmid19808876}

\bibitem{pmid16954142}
Roshan U, Livesay D (2006) {{P}robalign: multiple sequence alignment using
  partition function posterior probabilities}.
\newblock Bioinformatics 22: 2715--2721.
\bibAnnoteFile{pmid16954142}

\bibitem{DBLP:conf/recomb/DoGB06}
Do CB, Gross SS, Batzoglou S (2006) Contralign: Discriminative training for
  protein sequence alignment.
\newblock In: Apostolico A, Guerra C, Istrail S, Pevzner PA, Waterman MS,
  editors, RECOMB. Springer, volume 3909 of \emph{Lecture Notes in Computer
  Science}, pp. 160-174.
\bibAnnoteFile{DBLP:conf/recomb/DoGB06}

\bibitem{pmid10329189}
Mathews DH, Sabina J, Zuker M, Turner DH (1999) {{E}xpanded sequence dependence
  of thermodynamic parameters improves prediction of {R}{N}{A} secondary
  structure}.
\newblock J Mol Biol 288: 911--940.
\bibAnnoteFile{pmid10329189}

\bibitem{pmid15180907}
Dowell R, Eddy S (2004) {{E}valuation of several lightweight stochastic
  context-free grammars for {R}{N}{A} secondary structure prediction}.
\newblock BMC Bioinformatics 5: 71.
\bibAnnoteFile{pmid15180907}

\bibitem{pmid12912839}
Ronquist F, Huelsenbeck JP (2003) {{M}r{B}ayes 3: {B}ayesian phylogenetic
  inference under mixed models}.
\newblock Bioinformatics 19: 1572--1574.
\bibAnnoteFile{pmid12912839}

\bibitem{pmid11524383}
Huelsenbeck JP, Ronquist F (2001) {{M}{R}{B}{A}{Y}{E}{S}: {B}ayesian inference
  of phylogenetic trees}.
\newblock Bioinformatics 17: 754--755.
\bibAnnoteFile{pmid11524383}

\bibitem{Durbin:1998}
Durbin R, Eddy S, Krogh A, Mitchison G (1998) Biological sequence analysis.
\newblock Cambridge, UK: Cambridge University press.
\bibAnnoteFile{Durbin:1998}

\bibitem{pmid5420325}
Needleman S, Wunsch C (1970) {{A} general method applicable to the search for
  similarities in the amino acid sequence of two proteins}.
\newblock J Mol Biol 48: 443--453.
\bibAnnoteFile{pmid5420325}

\bibitem{pmid20110255}
Frith MC, Wan R, Horton P (2010) {{I}ncorporating sequence quality data into
  alignment improves {D}{N}{A} read mapping}.
\newblock Nucleic Acids Res 38: e100.
\bibAnnoteFile{pmid20110255}

\bibitem{pmid9773345}
Holmes I, Durbin R (1998) {{D}ynamic programming alignment accuracy}.
\newblock J Comput Biol 5: 493--504.
\bibAnnoteFile{pmid9773345}

\bibitem{pmid15608160}
Griffiths-Jones S, Moxon S, Marshall M, Khanna A, Eddy SR, et~al. (2005) {Rfam:
  annotating non-coding RNAs in complete genomes}.
\newblock Nucleic Acids Res 33: 121--124.
\bibAnnoteFile{pmid15608160}

\bibitem{pmid17646296}
Andronescu M, Condon A, Hoos H, Mathews D, Murphy K (2007) {{E}fficient
  parameter estimation for {R}{N}{A} secondary structure prediction}.
\newblock Bioinformatics 23: 19--28.
\bibAnnoteFile{pmid17646296}

\bibitem{Nussinov1978}
Nussinov R, Pieczenk G, Griggs J, Kleitman D (1978) {Algorithms for loop
  matchings}.
\newblock SIAM Journal of Applied Mathematics 35: 68--82.
\bibAnnoteFile{Nussinov1978}

\bibitem{pmid12824337}
Zuker M (2003) {Mfold web server for nucleic acid folding and hybridization
  prediction}.
\newblock Nucleic Acids Res 31: 3406--3415.
\bibAnnoteFile{pmid12824337}

\bibitem{hofacker1994}
Hofacker I, Fontana W, Stadler P, Bonhoeffer S, Tacker M, et~al. (1994) {Fast
  folding and comparison of RNA secondary structures.}
\newblock Monatsh Chem 125: 167--188.
\bibAnnoteFile{hofacker1994}

\bibitem{pmid15123812}
Mathews D, Disney M, Childs J, Schroeder S, Zuker M, et~al. (2004)
  {{I}ncorporating chemical modification constraints into a dynamic programming
  algorithm for prediction of {R}{N}{A} secondary structure}.
\newblock Proc Natl Acad Sci USA 101: 7287--7292.
\bibAnnoteFile{pmid15123812}

\bibitem{pmid16109749}
Chan CY, Lawrence CE, Ding Y (2005) {{S}tructure clustering features on the
  {S}fold {W}eb server}.
\newblock Bioinformatics 21: 3926--3928.
\bibAnnoteFile{pmid16109749}

\bibitem{pmid15215366}
Ding Y, Chan CY, Lawrence CE (2004) {Sfold web server for statistical folding
  and rational design of nucleic acids}.
\newblock Nucleic Acids Res 32: 135--141.
\bibAnnoteFile{pmid15215366}

\bibitem{pmid19833701}
Bernhart SH, Hofacker IL (2009) {{F}rom consensus structure prediction to
  {R}{N}{A} gene finding}.
\newblock Brief Funct Genomic Proteomic 8: 461--471.
\bibAnnoteFile{pmid19833701}

\bibitem{pmid19369331}
Schroeder SJ (2009) {{A}dvances in {R}{N}{A} structure prediction from
  sequence: new tools for generating hypotheses about viral {R}{N}{A}
  structure-function relationships}.
\newblock J Virol 83: 6326--6334.
\bibAnnoteFile{pmid19369331}

\bibitem{pmid19723687}
Stocsits RR, Letsch H, Hertel J, Misof B, Stadler PF (2009) {{A}ccurate and
  efficient reconstruction of deep phylogenies from structured {R}{N}{A}s}.
\newblock Nucleic Acids Res 37: 6184--6193.
\bibAnnoteFile{pmid19723687}

\bibitem{Washietl-2005}
Washietl S, Hofacker IL, Stadler PF (2005) {Fast and reliable prediction of
  noncoding RNAs}.
\newblock Proc Natl Acad Sci U S A 102: 2454--2459.
\bibAnnoteFile{Washietl-2005}

\bibitem{pmid16273071}
Washietl S, Hofacker IL, Lukasser M, Huttenhofer A, Stadler PF (2005) {Mapping
  of conserved RNA secondary structures predicts thousands of functional
  noncoding RNAs in the human genome}.
\newblock Nat Biotechnol 23: 1383--1390.
\bibAnnoteFile{pmid16273071}

\bibitem{pmid19908361}
Okada Y, Sato K, Sakakibara Y (2010) {{I}mprovement of structure conservation
  index with centroid estimators}.
\newblock Pac Symp Biocomput : 88--97.
\bibAnnoteFile{pmid19908361}

\bibitem{pmid19014431}
Bernhart S, Hofacker I, Will S, Gruber A, Stadler P (2008) {{R}{N}{A}alifold:
  improved consensus structure prediction for {R}{N}{A} alignments}.
\newblock BMC Bioinformatics 9: 474.
\bibAnnoteFile{pmid19014431}

\bibitem{sankoff1985}
Sankoff D (1985) {Simultaneous solution of the RNA folding alignment and
  protosequence problems}.
\newblock SIAM J Appl Math : 810--825.
\bibAnnoteFile{sankoff1985}

\end{thebibliography}



%

%

%

%
%




\end{document}